%% file: main.tex
\documentclass[a4paper,11pt]{article}
\pdfoutput=1

\usepackage{tcs}

\begin{document}

\title{Krylov Methods are (nearly) Optimal for \\ Low-Rank Approximation}
\author{
Ainesh Bakshi \\
\texttt{ainesh@mit.edu} \\
MIT
\and
Shyam Narayanan \\
\texttt{shyamsn@mit.edu} \\
MIT
}
\date{}

\maketitle

\begin{abstract} 

We consider the problem of rank-$1$ low-rank approximation (LRA) in the matrix-vector product model under various Schatten norms: 
\begin{equation*}
    \min_{ \norm{u}_2=1} \norm{A (I - u u^\top)}_{\calS_p} ,
\end{equation*}
where $\norm{M}_{\calS_p}$ denotes the $\ell_p$ norm of the singular values of $M$. Given $\eps>0$, our goal is to output a unit vector $v$ such that 
\begin{equation*}
    \norm{ A\Paren{I - vv^\top} }_{\calS_p} \leq \Paren{1+\eps}\min_{\norm{u}_2=1}\norm{A\Paren{I - u u^\top}}_{\calS_p}.
\end{equation*}
Our main result shows that Krylov methods (nearly)
achieve the information-theoretically optimal\footnote{For Spectral LRA, the upper and lower bounds match up to a fixed universal constant. For Frobenius and Nuclear LRA, they match up to a $\log(1/\eps)$ factor. } number of matrix-vector products for  
Spectral ($p=\infty$), Frobenius ($p=2$) and Nuclear ($p=1$) LRA.

In particular, for Spectral LRA, we show that any algorithm requires $\Omega\Paren{ \log(n)/\eps^{1/2} }$ matrix-vector products, exactly matching the upper bound obtained by Krylov methods~\cite{musco2015randomized}. Our lower bound addresses Open Question 1 in~\cite{woodruff2014sketching}, providing evidence for the lack of progress on algorithms for Spectral LRA and resolves Open Question 1.2 in~\cite{bakshi2022low}.  Next, we show that for any fixed constant $p$, i.e. $1\leq p =O(1)$, there is an upper bound of $O\Paren{ \log(1/\eps)/\eps^{1/3} }$ matrix-vector products, implying that the complexity does not grow as a function of input size. This improves the $O\Paren{ \log(n/\eps)/\eps^{1/3} }$ bound recently obtained in~\cite{bakshi2022low}, and matches their $\Omega\Paren{1/\eps^{1/3}}$ lower bound, to a $\log(1/\eps)$ factor. 
\end{abstract}

\thispagestyle{empty}

\clearpage

\microtypesetup{protrusion=false}
\tableofcontents{}
\thispagestyle{empty}
\microtypesetup{protrusion=true}

\clearpage

\pagestyle{plain}
\setcounter{page}{1}

\input{intro}

\input{krylov_lb_new}
\input{related-work}

\input{prelim}
\input{block_krylov}
\input{lifting-result}

\input{faster-krylov-algorithms}

\section*{Acknowledgements}

The authors thank Sinho Chewi, Jaume de Dios Pont, Piotr Indyk, Jerry Li, Chen Lu, and Erik Waingarten for helpful discussions. AB is supported by Ankur Moitra’s ONR grant. 
SN is supported by the NSF TRIPODS Program, an NSF Graduate Fellowship, and a Google Fellowship.


\input{biblio}

\appendix

\input{lifting-proof}

\end{document}

%% file: intro.tex
\section{Introduction}

Iterative algorithms are the workhorse of modern optimization methods and are pervasive throughout scientific computing, numerical linear algebra and machine learning. 
Such algorithms are now used for a wide array of tasks, from training large machine learning models~\cite{boyd2004convex, goodfellow2016deep}, running large-scale simulations for fluid dynamics~\cite{elman1996iterative,elman1996fast,elman1996multigrid}, structural analysis~\cite{roux1989acceleration, toselli2004domain}  and computational chemistry~\cite{schlick2009optimization,zuev2015new}, to quantum machine learning~\cite{harrow2009quantum,gilyen2019quantum}. In order to develop a general theory of iterative algorithms, and systematically compare their performace, we need a computational model that simulatenously captures all such algorithms and admits fine-grained lower bounds. While the standard RAM model easily captures iterative algorithms, we have no tools obtain fine-grained lower bounds on their performace.

An alternate computational model that has recieved significant attention lately is the \textit{matrix-vector product} model~\cite{WWZ14, SunWYZ19,RWZ20,SAR18,braverman2020gradient,MMMW21, bakshi2022low,needell2022testing}. Here, the algorithm accesses an input matrix $A$ only via adaptive matrix-vector queries. In particular, the algorithm chooses
a query vector $v^1$, obtains the product $A \cdot v^1$, chooses the next query
vector $v^2$, which is any randomized function of $v^1$ and $A \cdot v^1$, receives $A \cdot v^2$,
and so on. The fundametal measure of complexity in this model is the minimum number of matrix-vector products required to solve a given problem, which we refer to as the \textit{matrix-vector complexity}. 

The \textit{matrix-vector product} model captures natural iterative algorithms and has been extensively studied in the scientific computing and numerical linear algebra communities (see, for instance, ~\cite{knoll2004jacobian}, and references therein). Further, in many real-world applications, the number of matrix-vector products dominate the overall running time~\cite{mellor2004optimizing}. Finally, it is possible to obtain unconditional, information-theoretical lower bounds on the \textit{matrix-vector} complexity for various problems, as demonstrated by~\cite{SAR18,braverman2020gradient}, for computing the top eigenvalue of a matrix.

A popular class of iterative algorithms are based on computing the Krylov subspace: we loosely refer to such algorithms as \textit{Krylov subspace methods}. Here, instead of discarding intermediate matrix-vector products, the algorithm constructs a basis for the subspace spanned by intermediate vectors, i.e. $\calK = [v, A v, A^2 v, \ldots A^t v]$. 
Canonical examples of \textit{Krylov subspace methods} include Krylov iteration (Algorithm~\ref{algo:block-krylov-single-vector}) to compute top-$k$ eigenvalues and low-rank approximations~\cite{rokhlin2010randomized, halko2011finding,  musco2015randomized}, Conjugate Gradient to solve a linear system, 
 and Lanczos iteration to apply low-degree polynomials to eigenvalues (see ~\cite{saad1981krylov} and references therein).

In this work, we focus on understanding the \textit{matrix-vector complexity} of low-rank approximation, in the special case where the target rank is $1$. In particular, given an $n \times d$ matrix $A$ and accuracy parameter $0<\eps<1$,  the goal is to compute a unit vector $v$ such that 
\begin{equation*}
    \Norm{ A \Paren{I - vv^\top} }_{\calS_p} \leq  \min_{\norm{u}_2 =1 } \Paren{1+\eps}\ \Norm{ A \Paren{I - uu^\top} }_{\calS_p},
\end{equation*}
where $\Norm{M}_{\calS_p}$ is the Schatten-$p$ norm of $M$, defined as the $\ell_p$ norm of the singular values of $A$. Formulating low-rank approximation under Schatten-$p$ norms provides a convenient way to compare algorithms for well-studied matrix norms: Spectral $(p=\infty)$, Frobenius $(p=2)$, and Nuclear $(p=1)$. We note that any lower bound for rank-$1$ LRA implies a lower bound when the rank is a fixed universal constant, and our upper bounds extend naturally to the rank-$k$ approximation setting. For ease of exposition, we focus on the rank-$1$ LRA problem.  

Recently, Bakshi, Clarkson and Woodruff~\cite{bakshi2022low} studied Krylov methods in the matrix-vector product model for Schatten-$p$ low-rank approximation. They obtained an upper bound of $O\Paren{p^{1/6}\log(n/\eps) /\eps^{1/3} }$ matrix-vector products for any $p\geq 1$ by exploiting a trade-off between iterations and \textit{block size} (the number of starting vectors that are multiplied by $A$ in each step). They
construct two indepedent Krylov subspaces, $\calK_1 = [g , Ag, A^2 g, \ldots]$ and $\calK_2 = [G, AG, A^2 G, \ldots]$, where $\calK_2$ starts with a block matrix instead of a single vector.~\cite{bakshi2022low} also obtain an $\Omega\Paren{1/\eps^{1/3}}$ lower bound for any $p$ that is a fixed universal constant. On the other hand,  for Spectral LRA, the gap free analysis of Krylov Iteration by Musco and Musco~\cite{musco2015randomized} obtains  a $O\Paren{ \log(n) /\eps^{1/2} }$ upper bound.  Further, to the best of our knowledge, there is no known matrix-vector lower bound for Spectral low-rank approximation (see Section 1.2~\cite{bakshi2022low}, which explicitly states this as an open question).

Therefore, a natural question to ask is as follows:

\begin{center}
{\it Does Krylov iteration  achieve the optimal number of matrix-vector products for Spectral, Frobenius and Nuclear low-rank approximation? }
\end{center}

\subsection{Our Results}

\begin{table}[hbt]
\label{table:comparsion}
\centering
\begin{tabular}{|c|c|c|c|}
\hline
Reference                                                                   & Spectral $(p=\infty)$ & Frobenius $(p=2)$ & Nuclear $(p=1)$\\ \hline
\begin{tabular}[c]{@{}c@{}}Simultaneous Iteration\\ \cite{rokhlin2010randomized, halko2011finding,  musco2015randomized}\end{tabular} & $O\Paren{ \log(n)/\eps }$    & $O\Paren{ \log(n)/\eps }$      &  N.A.   \\ \hline
\begin{tabular}[c]{@{}c@{}}Block Krylov\\ ~\cite{musco2015randomized}\end{tabular}           & $O\Paren{\log(n)/ \eps^{1/2}} $     & $O\Paren{\log(n)/ \eps^{1/2}} $       & N.A.    \\ \hline
\begin{tabular}[c]{@{}c@{}}Modified Block Krylov\\ ~\cite{bakshi2022low}\end{tabular} & $O\Paren{ \log(n/\eps)/\eps^{1/2} }$     & $O\Paren{ \log(n/\eps)/\eps^{1/3} }$      & $O\Paren{ \log(n/\eps)/\eps^{1/3} }$    \\ \hline
\begin{tabular}[c]{@{}c@{}} Prior Lower Bounds\\ ~\cite{bakshi2022low}\end{tabular} & N.A      & $\Omega\Paren{ 1/\eps^{1/3} }$      & $\Omega\Paren{1/\eps^{1/3} }$    \\ \hline
\begin{tabular}[c]{@{}c@{}} Our Results \\
Thm~\ref{thm:spectral-lra-lower-bound}, Thm~\ref{thm:upper-bound-schatten-p} \end{tabular} & $\Omega\Paren{\log(n)/\eps^{1/2}}$      & $O\Paren{ \log(1/\eps)/\eps^{1/3} }$      & $O\Paren{ \log(1/\eps)/\eps^{1/3} }$     \\ \hline
\end{tabular}
\caption{Comparison of our results with prior work, measuring number of  matrix-vector products for Spectral, Frobenius and Nuclear low-rank approximation. We note that our lower bound implies Block Krylov~\cite{musco2015randomized} is the optimal algorithm for Spectral low-rank approximation. Our upper bound for Frobenius and Nuclear low-rank approximation matches the lower bound from~\cite{bakshi2022low} up to a $\log(1/\eps)$ factor.  }
\end{table}

We answer the aforementioned question in the affirmitive and show that Krylov iteration with a single starting vector (Algorithm~\ref{algo:block-krylov-single-vector}) obtains (nearly) optimal matrix-vector products for Spectral, Frobenius and Nuclear LRA. For Spectral LRA, the matrix-vector complexity is $\Theta\Paren{\log(n)/\sqrt{\eps} }$. For Frobenius and Nuclear LRA, the matrix-vector complexity is $\tilde{\Theta}\Paren{ 1/ \eps^{1/3} }$, where $\tilde{\Theta}$ surpresses a single $\log(1/\eps)$ factor (see Table~\ref{table:comparsion} for explicit upper and lower bounds). 

We begin by stating our lower bound for Spectral low-rank approximation:

\begin{theorem}[Lower Bound for Spectral LRA]
\label{thm:spectral-lra-lower-bound}
There exists a distribution $\calD$ over symmetric real $n\times n$ matrices such that given $A \sim\calD$ and $0< \eps<1$, any randomized algorithm requires $\Omega\Paren{\log(n)/\eps^{1/2}}$ matrix-vector products to output a vector $v$ such that with probability at least $2/3$,
\begin{equation*}
    \norm{A\Paren{I - vv^\top} }_{\textrm{op}} \leq \Paren{1+\eps} \min_{\norm{u}_2=1 }\norm{ A\Paren{I - uu^\top} }_{\textrm{op}}.
\end{equation*}
\end{theorem}

\begin{remark}[On Optimality]
Krylov Iteration~\cite{musco2015randomized} needs  $O\Paren{  \log(n)/\eps^{1/2} }$ matrix-vector products,  and therefore we resolve the matrix-vector complexity of rank-$1$ Spectral low-rank approximation. 
\end{remark}

\begin{remark}[Matrix-Vector vs. RAM]
In the RAM model,  Krylov iteration can be implemented in $O\Paren{\nnz(A)\log(n)/\eps^{1/2} }$ time, and Open Problem 1 in Woodruff's monograph~\cite{woodruff2014sketching} asks whether this can be improved to $O\Paren{\nnz(A) + n\poly(1/\eps) }$. This question has also been restated in several recent papers~\cite{bakshi2021learning,kacham2021reduced,woodruff2022improved}. Our result provides evidence for the lack of algorithmic progress on this problem.    
\end{remark}

\begin{remark}[Comparison to Prior Work]
To the best of our knowledge there is no known matrix-vector lower bound for Spectral LRA. Simchowitz, Alaoui and Recht~\cite{SAR18} obtain a matrix-vector lower bound for estimating the top-$k$ eigenvalues. However this does not translate to any lower bound for Spectral LRA (see Appendix A in~\cite{bakshi2022low} for details). Braverman, Hazan, Simchowitz and Woodworth~\cite{braverman2020gradient} introduce a different hard instance for estimating the top eigenvalue, and Bakshi, Clarkson and Woodruff adapt this instance to a lower bound of $\Omega\Paren{1/\eps^{1/3}}$ for Schatten-$p$ LRA, when $p$ is a fixed constant. In fact, they leave obtaining any lower bound that grows as a function of $1/\eps$, for Spectral LRA as an open problem (see Section 1.2 in~\cite{bakshi2022low}).    
\end{remark}


Next, we state our upper bound for Schatten-$p$ low-rank approximation, for any $p\geq 1$ that is bounded by a fixed constant.

\begin{theorem}[Upper bound for Schatten-$p$ LRA]
\label{thm:upper-bound-schatten-p}
Given a $n\times d$ matrix $A$, $0< \eps<1$, and $1\leq p$, there exists an algorithm that requires $O\Paren{p\log(1/\eps)/\eps^{1/3}}$ matrix-vector products and outputs a unit vector $v$ such that with probability at least $99/100$,
\begin{equation*}
    \Norm{ A\Paren{I - vv^\top} }_{\calS_p} \leq \Paren{1+\eps} \min_{\norm{u}_2=1} \Norm{ A \Paren{I - uu^\top} }_{\calS_p}.
\end{equation*}
\end{theorem}

\begin{remark}[On Optimality]
Bakshi, Clarkson and Woodruff~\cite{bakshi2022low} obtained an upper bound of $O\Paren{\log(n/\eps)p^{1/6}/\eps^{1/3}}$ and a lower bound of $\Omega\Paren{1/\eps^{1/3}}$ for a fixed constant $p$. 
When $p$ is a fixed universal constant (which holds for Frobenius and Nuclear LRA), we match the aforementioned lower bound up to a $\log(1/\eps)$ factor. 
\end{remark}

\begin{remark}[Krlov Iteration vs. Block Krylov]
Bakshi, Clarkson and Woodruff run two instantiations of Krylov methods in parallel, one with \textit{block size} $1$ and another with \textit{block size} $O(1/\eps^{1/3})$. Our algorithm only requires one instantiation, with \textit{block size} $1$, which is known to be more numerically stable in practice~\cite{carson2021stability}. Combined with our lower bounds, our results imply that we never need to run Krylov iteration with a \textit{block size} larger than $1$ for rank-$1$ low-rank approximation.
\end{remark}

We believe our algorithm and analysis may be generalizable to rank-$k$ low-rank approximation as well,
but we focus on the rank-$1$ case for ease of exposition.

Finally, we highlight a \textit{lifting theorem} that shows under fairly general conditions, a lower bound against Krylov iteration with large block size translates to an information-theoretic lower bound against adaptive queries in the matrix-vector model, which may be of independent interest.
This lifting result comes from a yet-unpublished paper of~\cite{samplingLB}, whose proof is included in Appendix \ref{appendix:lifting} for completeness.\footnote{%
We have received explicit permission from the authors of~\cite{samplingLB} to reproduce the proof.
}
We provide a simplified description of the lifting result here.

\begin{theorem}[Lifting Block Krylov Lower Bounds, Lemma~\ref{lem:chen_block_krylov} (informal)]
\label{thm:lifting-block-krylov}
    Let $\calA$ be any adaptive algorithm that makes $k$ matrix-vector queries to a $n \times n$ symmetric matrix $A$, where the eigenvectors of $A$ are uniformly (Haar-)random.  
     Then, given $k$ random Gaussians $V = [ v_1, \dots, v_k]$, one can perfectly simulate the distribution of $k$ adaptive queries and responses for $\calA$, given just the Krylov matrix $\calK = [ A V, A^2V , A^3V ,\ldots A^k V]$ and no other knowledge of $A$.
\end{theorem}

At a high-level, this result shows that if there exists any matrix-vector algorithm can solve a given eigenvalue or eigenvector problem with $k$ adaptive queries, then Krylov iteration with block size $k$ can be used to solve the same problem with $k$ iterations. In general, this simulator requires $k^2$ matrix-vector products, as opposed to $k$ matrix-vector products used by the adaptive algorithm, resulting in a quadratic overhead. However, in the setting of Spectral LRA, we demonstrate that starting with block size larger than $1$ adds no value (see Theorem~\ref{thm:lb-against-block-krylov-spectral-lra}), and the lifting technique obtains an optimal lower bound.

\begin{remark}
    We note~\cite{samplingLB} proved a slightly simpler version of \Cref{thm:lifting-block-krylov}, as they only require the lifting for solving a problem that depends on the eigenvalues (specifically, estimating $\Tr(A^{-1})$, for the purpose of generating a sample from $\calN(0, A^{-1})$). However, our lower bound for low-rank approximation is based on identifying good \emph{eigenvectors}, so \Cref{thm:lifting-block-krylov} is in fact a slight generalization of their result.
\end{remark}

\subsection{Open Problems}

We highlight the following open problems stemming from our work:

\begin{openquestion}[Larger Target Rank, refining Open Question 34 in~\cite{bakshi2022algorithms}]
While we focus on rank-$1$ low-rank approximation, the upper bound for Spectral low-rank approximation when the target rank is $k$ is $O\Paren{k\log(n)/\sqrt{\eps} }$ matrix-vector products~\cite{musco2015randomized}. Our lower bound implies that this is optimal for any target rank that is a fixed universal consant. What is the right \textit{matrix-vector} complexity as a function of $k$, $\log(n)$ and $1/\eps$ simultaneously?  
\end{openquestion}

\begin{openquestion}[Phase Transition for large $p$]
Perhaps surprisingly, when $p$ is a fixed universal consatant, the matrix-vector complexity does not grow, even logarithmically, with input size. However, the $\log(n)$ is neccesary for $p=\infty$ (see Theorem~\ref{thm:spectral-lra-lower-bound}). Is it possible to smoothly interpolate between these two regimes and obtain the correct dependence on $p$, $\log(n)$ and $1/\eps$ simultaneously? 
\end{openquestion}

%% file: krylov_lb_new.tex
\section{Technical Overview}
\label{sec:tech-overview}

In this section, we begin by describing our hard instance and provide a complete proof of a lower bound against Krylov Iteration (Algorithm~\ref{algo:block-krylov-single-vector}) for Spectral LRA. We then outline the approach to establish a lower bound against Block Krylov (Algorithm~\ref{algo:block-krylov-generalized}) and explain how to lift a lower bound against Block Krylov to a general matrix-vector product lower bound.
Finally, we discuss the new ideas we require to obtain a better upper bound for Schatten-$p$ LRA, when $p$ is a fixed constant.  

\subsection{Hard Instance and Lower Bound against Krylov Methods} \label{subsec:single-krylov-lb}

As a warm up, we provide a lower bound against Krylov Iteration (Algorithm~\ref{algo:block-krylov-single-vector}). At a high level, our proof proceeds by constructing a hard instance, $A$, such that the Krylov subspace $K = [g, Ag, A^2 g, \ldots, A^q g]$, for $q = c\log(n)/\sqrt{\eps}$, for a sufficiently small constant $c$, does not span any vector that has correlation at least $\sqrt{\eps/100}$ with the top eigenvector of $A$ (see Lemma~\ref{lem:krylov-alignment}). We then show that any $(1+\eps)$ relative-error Spectral low-rank approximation must have correlation at least $\sqrt{\eps/100}$ with the top-eigenvector of $A$ (see Lemma~\ref{lem:alignment-to-lra}).

Formally, we obtain the following theorem:

\begin{theorem}[Spectral LRA is hard for Krylov Methods]
\label{thm:lb-against-krylov-spectral-lra}
Given $\eps>0$, let $n = \Omega\Paren{1/\eps^2}$. Then, there exists a distribution over $n \times n$ matrices $A$ such that Algorithm \ref{algo:block-krylov-single-vector} requires $q= \Omega\Paren{ \log n/\sqrt{\eps}} $ matrix-vector products to output a unit vector $v$ such that with probability at least $99/100$,
\begin{equation*}
    \Norm{ A -  A v v^\top  }_{\op} \leq \Paren{1+\eps} \min_{ \norm{u}_2=1 } \Norm{ A - A uu^\top }_{\op}. 
\end{equation*}
\end{theorem}

\begin{mdframed}
  \begin{algorithm}[Krylov Iteration, \cite{musco2015randomized}]
    \label{algo:block-krylov-single-vector}\mbox{}
    \begin{description}
    \item[Input:] An $n \times n$ matrix $A$,  iteration count $q$. 
    
    \begin{enumerate}
    \item Let $g$ be a vector drawn  from $\mathcal{N}(0,I)$. Let $\calK = \left[ A  g ; A^2 g;  A^3 g; \ldots; A^q g\right]$ be the $n \times q$  Krylov matrix obtained by concatenating the vectors $Ag, \ldots, A^q g$. 
    \item Compute an orthonomal basis $Q$ for the column span of $\calK$. Let $M =  Q^\top  A^2 Q$. 
    \item Compute the top left singular vector of $M$, and denote it by $y_1$.
    \end{enumerate}
    \item[Output:] $v  = Q y_1$.  
    \end{description}
  \end{algorithm}
\end{mdframed}

\paragraph{Chebyshev Polynomial Background.} We begin by defining our main protagonist: Chebyshev polynomials.

\begin{definition}[Chebyshev Polynomials of the first kind.]
\label{def:chebyshev-polynomials}
For any $d \in \mathbb{N}$, the $d$-th Chebyshev polynomial is defined as follows:
\begin{equation*}
    T_d(x) = \begin{cases} \cos\Paren{ d \arccos(x) } \hspace{1.8in} \textrm{if  } \abs{x} \leq 1 \\
    \frac{1}{2} \Paren{ \Paren{ x - \sqrt{x^2 - 1 } }^d + \Paren{ x + \sqrt{x^2 - 1 } }^d  } \hspace{0.24in} \textrm{if } \abs{x} \geq 1
    \end{cases}
\end{equation*}
\end{definition}

It is easy to see from the above definition that Chebyshev polynomials are bounded when $\abs{x} \leq 1$ and oscillate between $-1$ and $1$. Therefore, we can define locations where they take extremal values:

\begin{definition}[Extrema of Chebyshev Polynomials]
\label{def:extrema-of-chebyshev}
For all $i \in [d]$,   $T_d(x_i)  \in \{-1, 1\}$ iff $x_i = \cos\Paren{i\cdot \pi / d}$.
\end{definition}






\paragraph{Hard Instance for Spectral LRA.}

Recall, our goal is to construct a hard instance such that a $t$-dimensional Krylov subspace does not span any vector that is $(\sqrt{\eps}/10)$-correlated with the top eigenvector, when $ q = c\log(n)/\sqrt{\eps}$, for a sufficiently small constant $c$. Intuitively, any vector in the Krylov subspace can be written as a random linear combination of a degree-$q$ polynomial applied to the eigenvalues of the input matrix. Therefore, we construct an instance where
\begin{enumerate}
    \item The top eigenvector is along a uniformly random direction and the top eigenvalue has gap of $\eps$ from the second largest eigenvalue (in magnitude).
    \item The location of the distinct eigenvalues (excluding the largest magnitude eigenvalue) is such that any bounded degree-$t$ polynomial must attain a large value (close to $1$ in magnitude) at at least one of the distinct eigenvalues.
    \item The algebraic multiplicity of each eigenvalue (excluding the top eigenvalue) is large enough that any polynomial that places non-trivial weight on a such an eigenvalue must imply that the corresponding vector in Krylov subspace is nearly uncorrelated with the top eigenvector.
\end{enumerate}

As alluded to earlier, we set the locations of the distinct eigenvalues to be the points where the degree-$(\log(n)/\sqrt{\eps})$ Chebyshev polynomial achieves extremal values (see Figure~\ref{fig:my_label}). We then duplicate each eigenvalue sufficiently many times, and set the eigenvectors to be a Haar random orthogonal matrix. 
In contrast to prior hard instances~\cite{SAR18,braverman2020gradient,bakshi2022low}, we design the eigenvalues of our hard instance (as opposed to picking a random matrix from a  Wishart ensemble or a deformed Wigner ensemble). Further, in our analysis, it is crucial for each eigenvalue (except the top one) to have high algebraic multiplicity, whereas in prior hard instances, each eigenvalue appeared with algebraic multiplicity exactly $1$.

We are now ready to formally define the hard instance. 


\begin{definition}[Hard Distribution]
\label{def:hard-distribution-chebyshev-polynomial-intro}
Given $\eps>0$, let $n \geq  1/\eps^2$. 
Let $A = U \Lambda U^\top$ be the eigen-decomposition, where $U$ in a uniformly random $n \times n$ matrix with orthonormal columns (see \Cref{def:haar-random}), and $\Lambda$ is a diagonal matrix of the eigenvalues. To define $\Lambda$, we first choose the top eigenvalue $\lambda_1 := 1+2 \eps$ to have multiplicity $1$. Next, we add $q+1$ additional distinct eigenvalues $\lambda_2, \dots, \lambda_{q+2},$ where $q =c \log(n)/\sqrt{\eps}$, so that $\lambda_{i+2} = \cos\left(\frac{i}{q} \cdot \pi\right)$ for $0 \le i \le q$. 
(Note: this $q$ is the same $q$ that we prove a Krylov iteration lower bound against.)
Each of these eigenvalues will have multiplicity $t$, where $t = \Theta(n\cdot \sqrt{\eps} / \log(n))$.
\end{definition}

\begin{figure}
    \centering
    \begin{tabular}{c | c}
\addheight{\includegraphics[width=75mm]{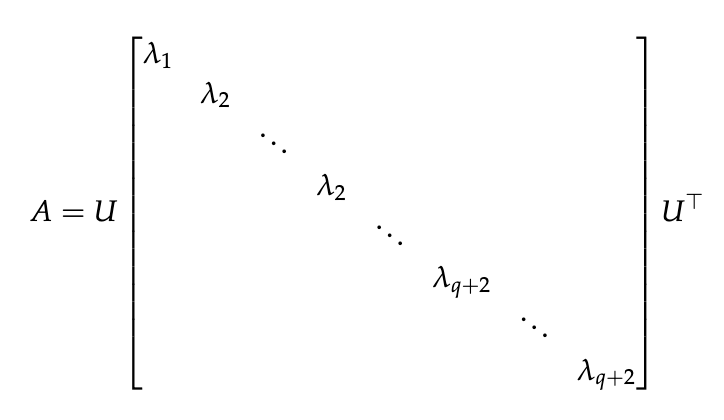}} & 
\addheight{\includegraphics[width=75mm]{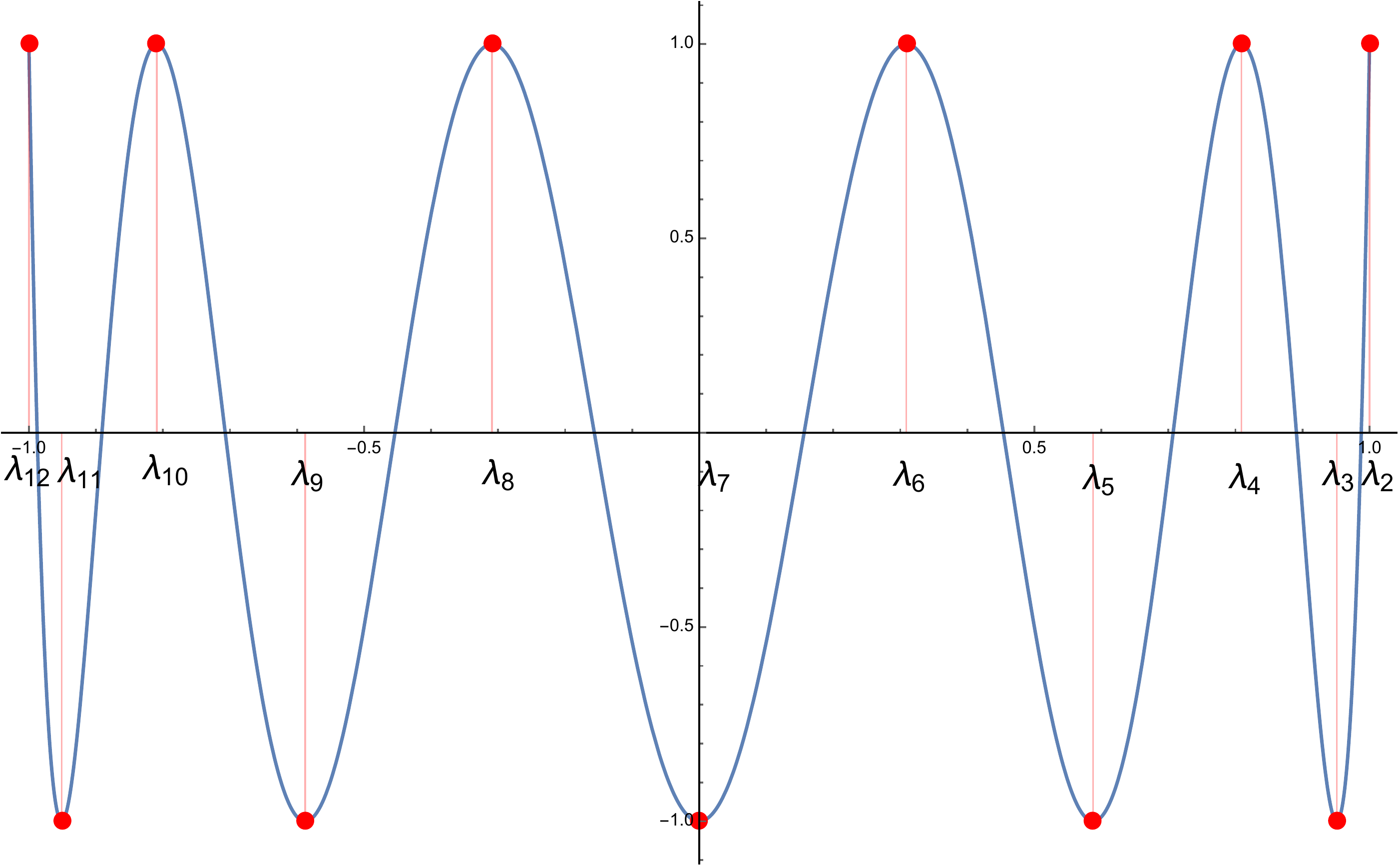}} \\
\end{tabular}
    \caption{ Let $n =\Omega(1/\eps^{2})$ and let $t= \Theta(n\sqrt{\eps}/\log(n)),$ and $q = \Theta(\log(n)/\sqrt{\eps})$. Let $\lambda_2, \lambda_3, \lambda_4,  \ldots, \lambda_{q+2}$ be the points that obtain the extrema of the degree-$q$ Chebyshev polynomial (see Definition~\ref{def:extrema-of-chebyshev}). We create a hard instance where we set $\lambda_1 = 1+2\eps$, and set the remaining eigenvalues to be $\lambda_i$ with algebraic multiplicity $t$, for each $i \in [2,q+2]$. Finally, set $A = U\Lambda U^\top$,  where $U$ is a  uniformly random orthogonal matrix. }
    \label{fig:my_label}
\end{figure}


The key lemma we establish shows that any vector in the Krylov subspace cannot be non-trivially correlated with the top eigenvector of $A$, unless the size of the Krylov subspace is $\Omega(\log(n)/\eps^{1/2})$. 

\begin{lemma}[Alignment of vectors in the Krylov Subspace]
\label{lem:krylov-alignment}
Given $g \sim \calN(0, I)$ and $\eps>0$,  let $A$ be sampled from the hard distribution in Definition \ref{def:hard-distribution-chebyshev-polynomial-intro}, let $\calK = [ g; Ag; A^2 g, \ldots , A^q g ]$ be the Krylov subspace for $q=  c \log(n)/\sqrt{\eps}$, for a sufficiently small constant $c$.  Further, let $u_1, u_2, \ldots , u_n$ be the eigenvectors of $A$, where $u_1$ is the top eigenvector, $u_2, \dots, u_{t+1}$ correspond to the second eigenvalue $\lambda_2$, and so on. Then, with probability at least $1-1/n$, for any vector $w$ in the column span of $\calK$, we have $\Iprod{ w, u_1 }^2 \leq  \eps/100$.  
\end{lemma}

\begin{proof}
We begin by recalling that we can rewrite $g = \sum_{i \in [n]}  a_i u_i$, for scalars $a_i = \Iprod{u_i, g}$, since the $u_i$'s span the entire space. Further, by rotational invariance of Gaussians, we know that each $a_i \overset{i.i.d.}{\sim} \calN(0,1)$. Now, we can group together all the eigenvectors that correspond to each eigenvalue with multiplicity $t$, where $t$ was set in Definition \ref{def:hard-distribution-chebyshev-polynomial-intro}. Formally, for $j \in [2, q+2]$, let 
\begin{equation*}
     \hat{u}_j = \sum_{ \ell \in [t] } \frac{  a_{(j-2) \cdot t + 1 + \ell} u_{(j-2) \cdot t + 1 +\ell}  } { \sqrt{ \sum_{ \ell \in [t] } a^2_{(j-2) \cdot t  +1+ \ell } } }  , 
\end{equation*}
be the \textit{average eigenvector} corresponding the $j$-th unique eigenvalue. Further, let 
\begin{equation*}
    \hat{a}_j = \sqrt{ \sum_{ \ell \in [t] } a^2_{(j-2)*t  +1+ \ell } }
\end{equation*}
be the corresponding coefficient in the expansion of $g$ in the eigenbasis. 
Then, we can rewrite $g$ as follows: 
\begin{equation*}
    g =  a_1 u_1 +  \sum_{ j \in [2,q+2] } \hat{a}_j \hat{u}_j . 
\end{equation*}
Recall the notation $\lambda_1, \lambda_2, \ldots , \lambda_{q+2}$ to denote the distinct eigenvalues of $A$. 
Observe that in the $r$-th iteration of Algorithm~\ref{algo:block-krylov-single-vector}, we obtain the vector 
\begin{equation}
\begin{split}
    A^r g & = U \Lambda^r U^\top \Paren{  a_1 u_1 +  \sum_{ j \in [2,q+2] } \hat{a}_j \hat{u}_j  }  \\
    & = \lambda_1^r a_1 u_1 + \sum_{j \in [2, q+2]} \lambda_j^r \hat{a}_j \hat{u}_j.
\end{split}
\end{equation}
After $q=c\log(n)/\sqrt{\eps}$ iterations, any vector in the Krylov subspace can be written as a linear combination of the columns and thus admits the following form:
\begin{equation}
\label{eqn:norm-of-w-intro}
\begin{split}
    w & = \sum_{r = 1}^{q}  c_r \Paren{ \lambda_1^r a_1 u_1 + \sum_{j \in [2, q+2]} \lambda_j^r \hat{a}_j \hat{u}_j  }  \\
    & = a_1 p(\lambda_1) u_1 
 + \sum_{j \in [2, q+2]}  p(\lambda_j) \hat{a}_j \hat{u}_j,
\end{split}
\end{equation}
where $p(x) = \sum_{r=1}^{q} c_r x^r$, for arbitrary scalers $c_r$. Further, by orthonormality of the eigenvectors, 

\begin{equation*}
\label{eqn:norm-of-w-2}
\norm{w}_2^2  =  a_1^2 p\Paren{\lambda_1}^2  + \sum_{ j \in [2, q+2]} p\Paren{\lambda_j}^2 \hat{a}_j^2.
\end{equation*}

Observe, $p$ is a degree $q$ polynomial and therefore has at most $q$ roots. By construction, $A$ has $\gg q$ distinct eigenvalues, and thus $p$ must obtain a  non-zero value on all but a small constant fraction of the  eigenvalues of $A$.
The main statement we show here is 
\begin{equation}
\label{eqn:main-inequality-p-at-roots}
     a_1^2 p(\lambda_1)^2 \leq \frac{ \eps}{100} \max_{j \in [2, q+2]}    \hat{a}_j^2 p(\lambda_j)^2, 
\end{equation}
where the $\lambda_j$'s denote the distinct eigenvalues of $A$. We first show how to complete the proof given Equation~\eqref{eqn:main-inequality-p-at-roots}. Consider the inner product of $w$ with the top eigenvector: 
\begin{equation}
\begin{split}
\Iprod{w, u_1}^2 & =   \Iprod{ a_1 p(\lambda_1) u_1 
 + \sum_{j \in [2, q+2]}  p(\lambda_j) \hat{a}_j \hat{u}_j, u_1}^2 \\
 &  = a_1^2 p(\lambda_1)^2 \norm{u_1}^2 \\
 & \leq  \frac{\eps }{100} \cdot \max_{j \in [2, q+2]} \hat{a}_j^2  p(\lambda_i)^2  \cdot \norm{u_1}^2 \\
 & \leq  \frac{\eps }{100}   \norm{w}_2^2,
\end{split}
\end{equation}
where the last inequality follows from Equation~\eqref{eqn:norm-of-w-intro}.

Therefore, it remains to prove Equation \eqref{eqn:main-inequality-p-at-roots}. 
First, consider the case where $\max_{i \in [2:q+2]} |p\Paren{\lambda_i}| = 0$. 
In this case $p$ has $q+1$ roots, but has degree at most $q$, so Equation~\eqref{eqn:main-inequality-p-at-roots} is trivially satisfied. Alternatively, since Equation~\eqref{eqn:main-inequality-p-at-roots} is scale invariant in $p$, without loss of generality, we may assume $\max_{i \in[2:q+2] } \abs{ p(\lambda_i) } = 1$, by scaling $p$ appropriately. By the fact that the degree $q$ Chebyshev polynomial $T_q$ is extremal, i.e. grows faster than any other bounded (at the values $\cos(\frac{i}{q} \cdot \pi)$) degree $q$ polynomial outside the interval $[-1,1]$ (see  Fact~\ref{fact:cheb-poly-extremal}), 
\begin{equation}
\label{eqn:bounding-lambda_1-v1}
\begin{split}
    \abs{ p(\lambda_1 ) } = \abs{ p( 1+2\eps ) }  & \leq T_{q}\Paren{1+2\eps}  \leq e^{C \min\Paren{ \sqrt{\eps} q, \eps q^2 } } \leq n^{0.01} \max_{i \in [2, q+2]} \abs{p(\lambda_i) }, 
\end{split}
\end{equation}
where the second inequality follows from standard bounds on the Chebyshev polynomial (see Fact~\ref{fact:growth-of-chebyshev}) and the last inequality follows from recalling our assumption that $\max_{i \in [2:q+2]} |p(\lambda_i)| = 1$, and that $q = c \log n/\sqrt{\eps}$.

Next, since $a_1\sim\calN(0,1)$, it follows from Fact~\ref{fact:deviation-of-Gaussian} that with probability at least $1-1/n^2$,  $\abs{a_1} \leq O(\sqrt{\log(n)})$. Further, for each $i \in [2, q+2]$, $\hat{a}_i$ is the Euclidean norm of a $t$-dimensional Gaussian vector. Using a standard concentration bound (see~Fact~\ref{fact:Gaussian-norm-conc}), we know that 
 with probability at least $1-1/2n$, for all $j \in [2, q+2]$,  
 \begin{equation*}
     \hat{a}_i = \Theta\Paren{\sqrt{t}} = \Theta\Paren{ \sqrt{ n \sqrt{\eps}/ \log(n)  } } \geq \frac{  c'  n^{1/2} \eps^{1/4} }{ \log(n) }  \abs{ a_1 }, 
 \end{equation*}
for a fixed constant $c'$. 
Union bounding over all $i \in [n]$, with probability at least $1-1/n$, for all $j \in [2, q+2]$, 
\begin{equation}
\label{eqn:bound-a1}
    \abs{ a_1 } \leq \frac{c' \log(n) }{  n^{1/2} \eps^{1/4} }  \hat{a}_j
\end{equation}
Combining \eqref{eqn:bounding-lambda_1-v1} and \eqref{eqn:bound-a1}, we have
\begin{equation*}
     \abs{a_1 p\Paren{\lambda_1 } }\leq \frac{c'  \log(n) n^{0.01} }{ n^{1/2} \eps^{1/4} } \max_{i \in [2,q+2]} \hat{a}_i \abs{ p\Paren{ \lambda_i } } \leq \frac{ \sqrt{ \eps} }{10}  \max_{i \in [2,q+2]} \hat{a}_i \abs{ p\Paren{ \lambda_i } } 
\end{equation*}
where the last inequality follows from recalling that $n \geq 1/\eps^2$, which concludes the proof of Equation~\eqref{eqn:main-inequality-p-at-roots}. 
\end{proof}

Next, we show that if any unit vector has small correlation (squared inner product less than $\eps$) with the top eigenvector, this vector cannot be a good Spectral low-rank approximation to $A$. 

\begin{restatable}{lemma}{uncorrelated}[Alignment to Spectral LRA]
\label{lem:alignment-to-lra}
Given $\eps>0$, let $A$ be a matrix such that $u_1$ is the top eigenvector of $A$. Further, let $\lambda_1 = 1+2 \eps$ and $\lambda_2 = 1$ be the top two eigenvalues of $A$, and suppose $|\lambda_i| \le 1$ for all $i \ge 2$.  Let $w$ be a vector such that $\Iprod{u_1, w}^2, \Iprod{u_1, Aw}^2 \leq \epsilon/2$. Then, 
\begin{equation*}
    \Norm{A\Paren{I - ww^\top} }_{\op} > \Paren{ 1+\eps } \min_{ \norm{u}_2=1 } \Norm{A\Paren{I - u u^\top}  }_{\op}.
\end{equation*}
\end{restatable}
\begin{proof}
By definition of the Operator norm, for any unit vector $v$, we have
\begin{equation*}
    \Norm{A\Paren{I - ww^\top} }_{\op} \geq  v^\top A\Paren{I - ww^\top} v .
\end{equation*}
In particular, for $u_1$, we have
\begin{equation*}
    \begin{split}
        u_1^\top A\Paren{I - ww^\top} u_1 &= u_1^\top A u_1 - \Iprod{u_1, Aw} \cdot \Iprod{ u_1, w } 
         \geq  1 + 2 \eps -\eps/2 > 1 + \eps .
    \end{split}
\end{equation*}
Further, $\min_{\norm{u}_2=1}\norm{A \Paren{I- 
 uu^\top}}_\op = 1$. Therefore, 
\begin{equation*}
    \Norm{A\Paren{I - ww^\top} }_{\op} \geq u_1^\top A (1-ww^\top) u_1 > 1+\eps = \Paren{ 1+\eps } 
 \min_{\norm{u}_2=1}\norm{A \Paren{I- 
 uu^\top}}_\op, 
\end{equation*}
which concludes the proof. 
\end{proof}

We are now ready to complete the proof of Theorem \ref{thm:lb-against-krylov-spectral-lra}.
\begin{proof}[Proof of Theorem~\ref{thm:lb-against-krylov-spectral-lra} ]
Observe, Lemma~\ref{lem:krylov-alignment} implies that with probability at least $99/100$,  any unit vector $w$ that is in the Krylov subspace (and in particular, the vector output by Algorithm~\ref{algo:block-krylov-single-vector}) satisfies $\Iprod{w, u_1}^2 \leq \eps/100$.
Further, $\Iprod{ u_1, A w }^2 = \Paren{1+2\eps}^2 \Iprod{ u_1, w }^2  \leq \eps/10 $. We invoke Lemma~\ref{lem:alignment-to-lra},
\begin{equation*}
    \Norm{  A \Paren{ I - ww^\top } }_{\op}  > \Paren{1 + \eps } \min_{\norm{u}_2=1 } \Norm{ A \Paren{I- uu^\top} }_{\op},
\end{equation*}
which concludes the proof.
\end{proof}

\subsection{Lifting Krylov Lower Bounds to Matrix-Vector Lower Bounds}

In \Cref{subsec:single-krylov-lb}, we described why a Krylov subspace with $c \log n/\sqrt{\eps}$ iterations cannot approximately locate the top eigenvector well enough to perform spectral LRA, for a sufficiently small constant $c$. Perhaps surpsingly, it turns out, that starting with a matrix instead of a single vector in Krylov iteration does not help. In particular, we show that even if we have $c \log n/\sqrt{\eps}$ starting vectors and run Krylov Iteration for $c \log n/\sqrt{\eps}$ iterations on each starting vector (Algorithm~\ref{algo:block-krylov-generalized}), the resulting Krylov subspace does not contain any vector that is even $\eps$-correlated with the top eigenvector of the input matrix.

In the single starting vector Krylov iteration lower bound, we showed that for any unit vector $w$ in the Krylov subspace, $|\langle u_1, w \rangle|$ is very small. Intuitively, this should imply that for block size $S$, a unit vector in the block Krylov subspace should not have norm more than $S \cdot |\langle u_1, w \rangle|$. Indeed, we can write $w = w_1 + w_2 + \cdots + w_S$, where each $w_i$ comes from the Krylov subspace generated by the $i$-th starting vector, so $|\langle u_1, w \rangle| \le \sum_{i=1}^S |\langle u_1, w_i \rangle|$.

The obstacle in excuting such an approach, however, is that the vectors $w_i$ may be anti-correlated and cancel out. As a result, it might be possible for some $w_i$ to have norm much bigger than $1$.  We are able to overcome this obstacle and show in our construction, with high probability, there is very little  anti-correlation between any potential $w_i$ and $w_j$ that are chosen. This insight allows for our single-vector Krylov iteration lower bound to extend to Block Krylov iteration.

Finally, we can apply the lifting result to show that any adaptive algorithm making $\log n/\sqrt{\eps}$ queries can be simulated by a Block Krylov algorithm. Intuitively, for an input instance where the eigenvectors are a uniformly Haar random matrix, the best an adaptive algorithm can do is explore a uniformly random direction in the complement of the subspace explored thus far. We make this intuition precise by showing that the sequence of adaptive queries can be modelled as $\Set{ v_i}_{i \in [k]}$, where $v_i = v_i^{\|} + v_i^{\bot}$, such that $v_i^{\|}$ is the component of $v_i$ in the span of the previous queries, and $v_i^{\bot}$ is orthogonal to this span. Further, the distributon of $v_i^{\bot}$ is uniformly random over the remaining subspace. We then show that Krylov iteration with large block size can simulate such adaptive queries and their responses (see Section~\ref{sec:lowerbounds-against-adaptive-algorithms} for details).



\subsection{Sharper Krylov Subspace Algorithms}

Finally, we show that we can improve the upper bound for Schatten-$p$ low-rank approximation obtained Bakshi, Clarkson and Woodruff~\cite{bakshi2022low}, when $p$ is bounded by a fixed constant. This includes the important special cases of Frobenius and Nuclear low-rank approximation.   
At a high-level, their algorithm instantiates two Krylov subspaces, one with a single starting vector that is iterated $O\Paren{\log(n/\eps)/\eps^{1/3}}$ times and another subspace with a starting block size of $O\Paren{ 1/\eps^{1/3} }$, iterated $O\Paren{\log(n/\eps)}$ times. Their analysis crucially relies on exploiting singular value gaps via large starting block size. 

Instead, we show that starting with a single starting vector, running Krylov Iteration (Algorithm~\ref{algo:block-krylov-single-vector}) for $t = O\Paren{ \log(1/\eps) /\eps^{1/3} }$ iterations converges to a unit vector $v$ such that 

\begin{equation*}
    \Norm{A \Paren{I - vv^\top} }_{\calS_p} \leq \Paren{1+\eps}\min_{\Norm{u}_2=1} \Norm{ A\Paren{ I - uu^\top } }_{\calS_p}.
\end{equation*}

For simplicity, we discuss the case of Frobenius norm low-rank approximation ($p = 2$);  the general case follows by replacing all instantiations of Pythagoreaon theorem with an appropriate generalization obtained in~\cite{bakshi2022low} (see Lemma~\ref{lem:schatten-pythagorean} for details). Further, for the purposes of the overview, we assume that
$\sigma_1 = 1$ (our final proof will never actually require knowledge of the spectral norm).

We perform a case anlysis similar to the one that appears in~\cite{bakshi2022low}. First, we consider the case where the top-$t$ singular values of $A$ are large and do not induce a gap, i.e. $\sum_{i \in [t]} \sigma_i^2 \geq  1/\eps^{1/3}$. In this case, we observe that the cost of the optimal solution itself is large:

\begin{equation*}
    \min_{ \norm{u}_2 = 1 } \norm{ A\Paren{I - uu^\top}  }_F^2 = \sum_{i = 2}^n \sigma_i^2 \ge \eps^{-1/3} -1 . 
\end{equation*}
A $(1+\eps)$ relative-error solution to the above cost corresponds to an \emph{additive} $\eps^{2/3}$ error. The standard analysis of Krylov iteration~\cite{musco2015randomized}, states that after $q= \log(n)/\sqrt{\zeta }$ iterations, for any $0< \zeta<1 $, the algorithm outputs a vector $v$ such that $v^\top A^\top A v \geq \norm{A}_{\textrm{op}}^2 - \zeta \sigma_2^2 $. By Pythagorean theorem, 
\begin{equation*}
    \begin{split}
        \norm{A\Paren{I - vv^\top}}_F^2 = \norm{A}^2_{\textrm{op}} - \norm{Av}_2^2 \leq \min_{ \norm{u}_2=1 } \norm{A\Paren{I - uu^\top} } + \sigma_2^2 \zeta.
    \end{split}
\end{equation*}
Since $\sigma_2^2 \leq 1$, it suffices to set $\zeta =\eps^{2/3}$ and thus $q = O( \log(n)/\eps^{1/3})$ iterations suffice. We strengthen this analysis by showing that a significantly lower dimensional Krylov subspace (corresponding to $O\Paren{\log(1/\eps)/\eps^{1/3}}$ iterations) spans a vector $v$ such that $\norm{A\Paren{I - vv^\top}}_F^2  \leq \min_{\norm{u}_2=1} \norm{A \Paren{I - uu^\top} } + \eps^{2/3}$.  We do this by explicity analyzing the Chebyshev polynomial (as opposed to a polynomial approximation to a threshold function in~\cite{musco2015randomized}) and demonstrate that the output of Algorithm~\ref{algo:block-krylov-single-vector} is at least as good as outputting the aformentioned vector $v$ (see Lemma~\ref{lem:existing-of-good-vectors} for details).

In the complementary case, we deviate significantly from any prior analysis of Krylov iteration, including~\cite{musco2015randomized, bakshi2022low}. Here, we know that the number of singular values in the range $[1/2, 1-\eps]$ is at most $O(\eps^{-1/3})$. We therefore construct an entirely different polynomial, which is no longer based on Chebyshev polynomials. This polynomial is designed to explicitly zero out all singular values in the range $[1/2, 1-\eps]$. We note that the degree of this polynomial, $p_0$, is only $O(\eps^{-1/3})$, and it allows us to remove the contribution of all medium sized singular values, similar to starting with a larger block size. However, it may still be the case that $p_0\Paren{\sigma_1}$ is significantly smaller than $p_0\Paren{\sigma_j}$, for some $\sigma_j$ outside the interval $[1/2, 1-\eps]$. 

To address this issue, we consider the polynomial $p_1(x) = x^q p_0(x)$, where $q= O\Paren{\log(1/\eps)/\eps^{1/3}}$. We show that since $p_1$ powers up the top sigular value significantly, $p_1\Paren{\sigma_1} \gg p_1\Paren{\sigma_j}$ for any $\sigma_j <1/2$. We then prove that the vector $v = p_1(A)g/\norm{p_1(A)g}$ results in a $(1+\eps)$ relative-error low-rank approximation and that Krylov iteration, after  $O\Paren{\log(1/\eps)/\eps^{1/3}}$ iterations, outputs a vector that does at least as well.


%% file: related-work.tex
\section{Additional Related Work}


In recent years, the matrix-vector product model has recieved considerable attention in the theoretical computer science community, since it was formalized for a number of problems 
in \cite{SunWYZ19,RWZ20}. 
Simultaneously, \cite{SAR18,braverman2020gradient} obtained nearly tight bounds for estimating the top eigenvector and eigenvalue. Next, for the problem of estimating the trace of a positive semidefinite matrix, tight bounds were obtained in \cite{MMMW21} (see, also~\cite{dharangutte2023tight} for estimating the diagonal). For recovering a planted clique
from a random graph, upper and lower bounds were obtained in \cite{woodruff21}. Finally,~\cite{bakshi2022low} studied the low-rank approximation problem and~\cite{needell2022testing} studied testing whether a matrix is PSD in the matrix-vector model.   

A closely related setting to the matrix-vector model is one where the input is accessed in a non-adaptive manner, i.e. $v^1, \ldots, v^q$, are chosen before making any queries to $A$. This model is equivalent to the {\it sketching model}, which is thoroughly studied on its own (see, e.g., \cite{nelson2011sketching,woodruff2014sketching}), and in the context of data streams \cite{M05,LNW14}. Low-rank approximation under Schatten norms has been well-studied in this model~\cite{clarkson2013low, lw20}.

Iterative methods, such as Krylov subspace based methods, are captured by the matrix-vector product framework, whereas linear sketching allows for the choice of a matrix $S \in \mathbb{R}^{t \times n}$, where $t$ is the number of ``queries'', and then observes the product $S \cdot A$ and so on (see~\cite{woodruff2014sketching} and references therein).  The model has important applications to streaming and distributed algorithms and several recent works have focused on estimating spectral norms and the top singular values~\cite{andoni2013eigenvalues,li2014sketching,li2016tight,bakshi2021learning}, estimating Schatten and Ky-Fan norms~\cite{li2016tight, li2017embeddings,li2016approximating, braverman2019schatten} and low-rank approximation~\cite{clarkson2013low,mm13,NN13,BDN15,cohen2016nearly}. 

Finally, the matrix-vector product model is also closely related to sublinear time/query algorithms and quantum-inspired algorithms. There has been a flurry of work on sublinear low-rank approximation under various structural assumptions on the input~\cite{mw17,  bakshi2018sublinear,indyk2019sample, sw19,bakshi2020robust}  and in quantum-inspired models~\cite{kerenidis2016quantum,  chia2018quantum, tang2019quantum, gilyen2018quantum,   gilyen2019quantum,chepurko2020quantum,bakshi2023improved}. 


%% file: prelim.tex
\section{Preliminaries}

Given an $n \times d$ matrix $A$ with rank $r$, and $n\geq d$, we can compute its 
singular value decomposition, denoted by ${SVD}(A) = U \Sigma V^{\top}$, such that $U$ is an $n \times r$ matrix with 
orthonormal columns, $V^{\top}$ is an $r \times d$ matrix with orthonormal 
rows and $\Sigma$ is an $r \times r$ diagonal matrix. The entries 
along the diagonal are the singular values of $A$, denoted by 
$\sigma_1, \sigma_2 \ldots \sigma_r$. Given an integer $k \leq r$, we 
define the truncated singular value decomposition of $A$ that zeros out 
all but the top $k$ singular values of $A$, i.e.,  $A_k = U 
\Sigma_k V^{\top}$, where $\Sigma_k$ has only $k$ non-zero 
entries along the diagonal. It is well-known that the truncated SVD 
computes the best rank-$k$ approximation to $A$ under any unitarily invariant norm, but in particular for any Schatten-$p$ norm (defined below), we have $A_k = \min_{ \rank(X)=k  } \| A - X \|_{\calS_p}$.
More generally, for any matrix $M$, we use the notation $M_k$ and 
$M_{\setminus k}$ to denote the first $k$ components and all but the 
first $k$ components respectively. 
We use $M_{i,*}$ and $M_{*,j}$ to 
refer to the $i^{th}$ row and $j^{th}$ column of $M$ respectively. 

We use the notation $I_k$ to denote a \emph{truncated identity matrix}, that is, a square matrix with its top $k$ diagonal entries equal to one, and all other entries zero. The dimension of $I_k$ will be determined by context.

 \paragraph{Schatten Norms.} We recall some basic facts for Schatten-$p$ norms. We also require the following trace and operator inequalities.

\begin{definition}[Schatten-$p$ Norm]
\label{def:schatten}
Given a matrix $A \in \mathbb{R}^{n \times d}$, let $\sigma_1 \geq \sigma_2 \geq \ldots \geq \sigma_d $ be the singular values of $A$. Then, for any $p \in [0, \infty)$, the Schatten-$p$ norm of $A$ is defined as 
\[
\norm{A }_{\calS_p} =\tr\Paren{ \Paren{A^\top A}^{p/2} }^{1/p} = \Paren{ \sum_{i \in [d]} \sigma_i^p(A) }^{1/p}.
\]
\end{definition}

\begin{fact}[Schatten-$p$ norms are Unitarily Invariant]
\label{fact:unitary_inv}
Given an $n\times d$ matrix $M$, for any $m\times n$ matrix $U$ with orthonormal columns, a norm $\|\cdot \|_X$ is defined to be unitarily invariant if $\|U M \|_X = \| M \|_X$. The Schatten-$p$ norm is unitarily invariant for all $p\geq1$. 
\end{fact}

There exists a closed-form expression for the low-rank approximation problem under Schatten-$p$ norms:

\begin{fact}[Schatten-$p$ Low-Rank Approximation]
\label{fact:schtatten_lra}
Given a matrix $A \in \mathbb{R}^{n \times d}$ and an integer $k \in \mathbb{N}$, 
\begin{equation*}
    A_k = \arg\min_{\rank( X) \leq k} \norm{A - X}_{\calS_p},
\end{equation*}
where $A_k$ is the truncated SVD of $A$.
\end{fact}

\paragraph{Chebyshev Polynomials.} Next, we recall some basic facts about Chebyshev polynomials.



\begin{definition}[Extrema of Chebyshev Polynomial~\cite{mason2002chebyshev}]
\label{def:extrema-chebyshev-polynomials}
The local extrema of the $d$-th Chebyshev polynomial, $T_d(x)$, are in the range $[-1, 1]$ and are given by 
\begin{equation*}
    x_i = \cos\Paren{ \frac{i}{d} \cdot \pi   },
\end{equation*}
for all $i \in [d-1]$. In addition, for every extrema $x_i$, $T_d(x_i) \in \{-1, 1\}$. Finally, the set of solutions to $T_d(x) \in \{-1, 1\}$ are given by
\begin{equation*}
    x_i = \cos\Paren{ \frac{i}{d} \cdot \pi   },
\end{equation*}
for all $i \in \{0, 1, \dots, d\}$.
\end{definition}

\begin{fact}[Chebyshev Polynomials are extremal~\cite{mason2002chebyshev}]
\label{fact:cheb-poly-extremal}
Let $q(x)$ be any degree-$d$ polynomial such that for all $x_i = \cos\left(\frac{i}{d} \cdot \pi\right)$ for $i \in \{0, 1, \dots, d\}$, $\abs{q(x_i)}\leq 1$. Then, for any $x > 1$, $\abs{q(x)}\leq T_d(x)$.  
\end{fact}

\begin{fact}[Growth of Chebyshev Polynomial~\cite{mason2002chebyshev}]
\label{fact:growth-of-chebyshev}
For any $0<\eps<0.5$,
\begin{equation*}
    e^{c \cdot \min(\sqrt{\eps} \cdot d, \eps \cdot d^2)} \le T_d(1+\eps)\leq e^{C \cdot \min(\sqrt{\eps} \cdot d, \eps \cdot d^2)},
\end{equation*}
for some fixed constants $C > c > 0$.
\end{fact}

The following is a well-known corollary of \Cref{fact:growth-of-chebyshev}

\begin{corollary} \label{cor:magic-polynomial}
    There exists a constant $c > 0$ such that for any $\eps < 0.5$, There exists a polynomial $T_{d, \eps}$ of degree $d$, such that $|T_{d, \eps}(x)| \le e^{-c \cdot \min(\sqrt{\eps} \cdot d, \eps \cdot d^2)}$ for all $0 \le x \le 1-\eps$, and $T_{d, \eps}(1) = 1$.
\end{corollary}

\begin{proof}
    Let $T_d$ be the degree $d$ Chebyshev polynomial. We define $T_{d, \eps}(x) := T_d(x+\eps)/T_d(1+\eps)$. Then, $T_{d, \eps}(1) = 1$. Moreover, for any $x \in [0, 1-\eps],$ $|T_{d, \eps}(x)| \le |T_d(x+\eps)|/|T_d(1+\eps)|$. However, $|T_d(x+\eps)| \le 1$ by definition of the Chebyshev polynomial (as $0 \le x \le 1-\eps$) and $|T_d(1+\eps)| \ge e^{c \cdot \min(\sqrt{\eps} \cdot d, \eps \cdot d^2)},$ where $c$ is the same constant as in \Cref{fact:growth-of-chebyshev}. So, $|T_{d, \eps}(x)| \le e^{-c \cdot \min(\sqrt{\eps} \cdot d, \eps \cdot d^2)}$ for all $0 \le x \le 1-\eps$.
\end{proof}

Also, as a direct corollary of Facts \ref{fact:cheb-poly-extremal} and \ref{fact:growth-of-chebyshev}, we have the following.

\begin{corollary} \label{cor:polynomial-growth}
    For any $0 < \eps < 0.5,$ and any polynomial $P$ of degree at most $d$,
\[P(1+\eps) \le e^{C \cdot \min(\sqrt{\eps} \cdot d, \eps \cdot d^2)} \cdot \max\limits_{i \in \{0, 1, \dots, d\}} \abs{P\Big(\cos\Big(\frac{i}{d} \cdot \pi\Big)\Big)}.\]
\end{corollary}

\begin{proof}
    Let $x_i := \cos\left(\frac{i}{d} \cdot \pi\right)$. First, if $P\left(x_i\right) = 0$ for all $i \in \{0, 1, \dots, d\}$, then $P$ has at least $d+1$ roots so $P \equiv 0$. 
    Alternatively, if $\max_{i \in \{0, 1, \dots, d\}} \abs{P(x_i)} = \kappa > 0$, then $q(x) = \frac{P(x)}{\kappa}$ satisfies $|q(x_i)| \le 1$ for all $i$, which means by \Cref{fact:cheb-poly-extremal}, $|q(1+\eps)| \le T_d(1+\eps)$. So, $|P(1+\eps)| \le \kappa \cdot T_d(1+\eps) \le \kappa \cdot e^{C \cdot \min(\sqrt{\eps} \cdot d, \eps \cdot d^2)}$, using \Cref{fact:growth-of-chebyshev}.
\end{proof}

\paragraph{Probability Background.}

We also require the following basic facts about probability distributions:

\begin{fact}[Deviation of a Gaussian]
\label{fact:deviation-of-Gaussian}
Let $g \sim \calN(0,1)$. Then, for any $\delta>0$, with probability at least $1-1/\delta$, $\abs{g} \leq \sqrt{\log(1/\delta)}$.
\end{fact}

\begin{fact}[Hanson-Wright]
\label{fact:Gaussian-norm-conc}
Let $g \sim \calN(0, I_d)$. Then, 
\begin{equation*}
\Pr\left[ \abs{ \norm{g}^2 - d } > t \right]  \leq 2 \exp\Paren{- c t}. 
\end{equation*}
\end{fact}

\begin{fact}[Singular Values of a Gaussian Matrix~\cite{vershynin2010introduction}]
\label{fact:singular-values-of-Gaussian}
Let $G\in \mathbb{R}^{n \times d}$ be such that for all $i\in [n], j\in[d]$, $G_{i,j} \sim\calN\Paren{0,1}$. Let $\sigma_1 \geq \sigma_2 \geq \ldots \geq \sigma_d$ be the singular values of $G$. Then, with probability at least $1-2\exp\Paren{-d/2}$, 
\begin{equation*}
    \sqrt{n} - 2\sqrt{d} \leq \sigma_d \leq \sigma_1 \leq \sqrt{n} + 2\sqrt{d}.
\end{equation*}
\end{fact}

We also require the following definitions:
\begin{definition}[Haar Random Matrix]
\label{def:haar-random}
Let $\mathbb{O}(n)$ be the orthogonal group on $n \times n$ matrices. There is a unique rotation invariant probability measure (Haar measure) $\mu$ on $\mathbb{O}(n)$. A Haar random matrix $U$ is a (matrix valued) sample from $\mu$.
\end{definition}

%% file: block_krylov.tex
\section{Lower Bound against Block Krylov Methods}

In this section, we show that increasing the block size in Krylov method based algorithms does not help to solve our hard instance from Section~\ref{sec:tech-overview}. In particular, we show that starting with a block size of $c\log(n)/\sqrt{\eps}$, for a small fixed constant $c$, and running $c\log(n)/\sqrt{\eps}$ iterations does not suffice to obtain a $(1+\eps)$-Spectral low-rank approximation. 


\vspace{0.1in}
\begin{mdframed}
  \begin{algorithm}[Block Krylov ``Algorithm'', generalization of \cite{musco2015randomized}]
    \label{algo:block-krylov-generalized}\mbox{}
    \begin{description}
    \item[Input:] An $n \times n$ matrix $A$, iteration count $r$, block size $s$. 
    
    \begin{enumerate}
    \item Let $g_1, \dots, g_s$ be vectors drawn i.i.d. from $\mathcal{N}(0,I)$. Let $\calK = \{A^t g_j\}_{0 \le t \le r, 1 \le j \le s}$ be the set of vectors. 
    \item Choose a unit vector $v \in \Span(\calK)$ to minimize $\left\|A (I - vv^\top) \right\|_{\op}$.
    \end{enumerate}
    \item[Output:] $A vv^\top$.  
    \end{description}
  \end{algorithm}
\end{mdframed}

\begin{theorem}[Spectral LRA is hard for Block Krylov Methods]
\label{thm:lb-against-block-krylov-spectral-lra}
Given $0<\eps<\frac{1}{2}$, let $n \ge \Omega\Paren{1/\eps^{2.01}}$. Then, there exists a distribution over $n \times n$ matrices $A$ and some small absolute constant $1>c > 0$ such that with probability at least $9/10$, for each vector $v$ in the Krylov subspace generated with $s =\frac{c \log n}{\sqrt{\eps}}$ random unit vectors, for $r =\frac{c \log n}{\sqrt{\eps}}$ iterations (see \Cref{algo:block-krylov-generalized}),
\begin{equation*}
    \Norm{ A \Paren{  I -  v v^\top } }_{\op} \geq (1+\eps) \min_{ \norm{u}_2=1 } \Norm{ A \Paren{I  -  uu^\top } }_{\op}.
\end{equation*}
\end{theorem}


Our hard distribution is essentially the same as before, but we restate it here for completeness. We construct a matrix $A = U D U^\top$, where $U$ is a uniformly random $n \times n$ orthogonal matrix, and $D$ is a fixed diagonal matrix of the eigenvalues. $D$ has  $O(\log(n)/\sqrt{\eps})$ distinct eigenvalues, where the top eigenvalue, $\lambda_1 = 1+2\eps$, has multiplicity $1$, and the remaining distinct eigenvalues, $\lambda_2, \lambda_3, \ldots , \lambda_{r+2}$ have multiplicity $O(n\sqrt{\eps}/\log(n))$. We set the remaining eigenvalues to be the distinct locations where the degree-$r$ Chebyshev polynomial $T_r$ equals $1$ or $-1$.

\begin{definition}[Hard Distribution]
\label{def:hard-distribution-chebyshev-polynomial}
Given $\eps>0$, let $n \in \mathbb{N}$ be such that $n= \Omega\Paren{1/\eps^{2.01}}$ and let $r = c \log(n)/\sqrt{\eps}$, for some fixed small constant $c$. Further, assume $k=\frac{n-1}{r+1}$ is an integer. 
For $t \in [2, r+2]$, let $\lambda_{t} := \cos\left(\frac{t-2}{R} \cdot \pi\right)$. Then, we define  
$D$ as the diagonal matrix with $D_{1, 1} = \lambda_1 = 1+2\eps$, and $D_{i, i} = \lambda_{1 + \lceil \frac{i-1}{k} \rceil}$ for each $2 \le i \le n$, so that every $\lambda_t$ for $2 \le t \le r+2$ is repeated exactly $k$ times. Finally, we define $A = U D U^\top$, where $U$ is a uniformly random $n \times n$ orthogonal matrix, and $D$ is the  fixed diagonal matrix defined above. 

\end{definition}

We note that $\lambda_2, \dots, \lambda_r$ are all in the range $[-1, 1]$, and in fact $\lambda_2 = 1$ and $\lambda_{r+2} = 1$.

The following fact is immediate by the Eckardt-Young theorem and the fact that $\max_{2 \le i \le r+2} |\lambda_i| = 1$.

\begin{fact}
    We have that $\min_{\|u\|_2 = 1} \norm{A-Auu^\top}_\op = 1$.
\end{fact}

We also recall the following result.
\uncorrelated*

We begin with the following lemma:

\begin{lemma}[Concentration Properties] \label{prop:concentration}
Let $A = V \Sigma V^\top$ be a sample from the hard instance in Definition~\ref{def:hard-distribution-chebyshev-polynomial}, such that $r=s=c\log(n)/\sqrt{\eps}$ and $k=(n-1)/(r+1)$. Let $g_1, g_2, \ldots, g_s \sim \calN(0, I)$ such that for each $j\in[s]$,  $g_j$ admits the following decomposition in the eigenbasis of $A$: $g_j = a_{j, 1}  v_{1} + \sum_{t \in [2, r+2]} a_{j, t} v_{t}$, where $v_1 = V_{*,1}$ and for any $t \in [2, r+2]$, $v_t = \sum_{i \in [k]} V_{*,(t-2)k +i +1}$.   With probability at least $0.9$, each of the following hold, assuming $k \ge 100 s$.
\begin{enumerate}
    \item For all $j \in [s]$, $a_{j, 1} \le 5 \sqrt{\log n}$.
    \item For all $j \in [s]$ and all $2 \le t \le r+2$, $0.5 \sqrt{k} \leq a_{j, t} \leq  2 \sqrt{k}$.  
    \item For all $2 \le t \le r+2$, define $V^{(t)} \in \BR^{n \times s}$ be the matrix with row $j$ equal to $v_{j, t}$. Then, for all $t \le r+2$, $V^{(t)}$ has all singular values between $\frac{1}{4}$ and $4$.
\end{enumerate}
\end{lemma}

\begin{proof}
    Since the top eigenvector has multiplicity $1$, $a_{j, 1}$ is the absolute value of a standard Gaussian $\calN(0, 1)$. So, Part 1 holds with probability at least $0.99$ by a union bound over $s \le n$ choices of $j$.

    For any $2 \le t \le r+2$, the $t^{\text{th}}$ eigenvector has multiplicity $k$, so the distribution of $a_{j, t}$ is the norm of a $k$-dimensional Gaussian. Since $k= \Theta\left(\frac{n}{r}\right)$ and $r \le \frac{\log n}{\sqrt{\eps}},$ if $n \gg \eps^{-1} \log^2 \eps^{-1}$ then $k \ge \sqrt{n}$. Hence, the probability that $0.5 \sqrt{k} \le a_{s, t} \le 2 \sqrt{k}$ is at least $1-e^{-\Omega(k)} \ge 1-e^{-\Omega(\sqrt{n})}$ by standard concentration inequalities. Taking a union bound over at most $n^2$ choices of $s$ and $t$, Part 2 holds with probability at least $0.99$.


    Since each $v_{j, t}$ has been normalized as a unit vector, we can write $V^{(t)} = G^{(t)} \cdot \diag(a_t)^{-1}$, where $G^{(t)}$ has column $s$ as $a_{j, t} v_{j, t}$, and $a_t$ is the vector $\{a_{j, t}\}_{1 \le j \le s}$. Assuming Part 2, $\diag(a_t)^{-1}$ has all singular values in the range $\left[\frac{0.5}{\sqrt{k}}, \frac{2}{\sqrt{k}}\right].$ Also, we can project each column of $G^{(t)}$ onto the $\lambda_t$-eigenspace, and then each column will be a random Gaussian vector in $k$ dimensions. Assuming that $k \ge 100 s,$ it follows from Fact~\ref{fact:singular-values-of-Gaussian} that all of the singular values of $G^{(t)}$ must lie in the range $\left[0.5 \sqrt{k}, 2 \sqrt{k}\right]$, with probability at least $1-\frac{1}{n^2}$. Hence, with at least $0.98$ probability, $V^{(t)}$ has all singular values between $\frac{1}{4}$ and $4$, for all $t$.
\end{proof}

From now on, we assume the three events in \Cref{prop:concentration} all hold, and use no other properties about the Gaussian vectors $g_s$. Due to \Cref{lem:alignment-to-lra}, it suffices to show the following lemma.

\begin{lemma} \label{lem:lb-against-krylov-main}
    Suppose the three events in \Cref{prop:concentration} hold. Then, for any unit vector $v \in \Span(\calK)$, $|\langle v, u_1 \rangle| < \sqrt{\frac{\eps}{10}}$.
\end{lemma}

\begin{proof}
We compute vectors $A^\ell g_j = \sum_{t = 1}^{r+2} \lambda_t^\ell \cdot a_{j, t} v_{j, t}$ for $j \le s, \ell  \le r+2$. Now, any $v \in \Span(\calK)$ must be a linear combination of $A^\ell g_j$ across $0 \le \ell \le r$ and $1 \le j \le s$. This means any such vector must be expressible in the form
\[\sum_{j = 1}^{s} \sum_{\ell = 0}^{r} p_{j, \ell } \cdot \sum_{t = 1}^{r+2} \lambda_t^\ell \cdot a_{j, t} v_{j, t} = \sum_{t = 1}^{r+2} \sum_{j = 1}^{s} a_{j, t} v_{j, t} \cdot \sum_{\ell = 0}^{r} p_{j, \ell} \lambda_t^\ell = \sum_{t = 1}^{r+2} \underbrace{\sum_{j = 1}^{s} a_{j, t} \phi_j(\lambda_t) v_{j, t}}_{x_t},\]
for some polynomials $\phi_1, \dots, \phi_s$ each of degree at most $r$. Also, note that $v_{j, 1} = \pm u_1$ for all $j$, as the eigenspace for $\lambda_1$ has dimension $1$.

First, consider some fixed $t \ge 2$. We will bound the norm of $x_t := \sum_{j = 1}^{s} a_{s, t} \phi_j(\lambda_t) v_{j, t}$.
Recalling that $V^{(t)} \in \BR^{n \times s}$ is the matrix with column $j$ equal to $v_{j, t}$, Part 3 of \Cref{prop:concentration} tells us that all singular values of $V^{(t)}$ are between $\frac{1}{4}$ and $4$. We can write $x_t = V^{(t)} \cdot w_t,$ where $w_t \in \BR^s$ has $j^\text{th}$ coordinate equal to $a_{j, t} \phi_j(\lambda_t)$. Therefore,
\begin{equation*}
    \|x_t\|_2^2 \ge \frac{1}{16} \cdot \sum_{j = 1}^{s} (a_{j, t} \phi_j(\lambda_t))^2 \ge \frac{k}{64} \cdot \sum_{j = 1}^{s} \phi_j(\lambda_t)^2.
\end{equation*}
The first inequality follows by Part 3 of \Cref{prop:concentration}, and the second inequality follows by Part 2 of \Cref{prop:concentration}.

In addition, we know that each $X_t$ is an eigenvector of eigenvalue $\lambda_t$, as every $v_{j, t}$ is. Therefore, they are orthogonal. This means that if $v = \sum_{t=1}^{r+2} X_t$ has norm $1$, then $\sum_{t=1}^{r+2} \|X_t\|_2^2 = 1$. However, note that
\begin{equation} \label{eq:Xnorm-bound}
    \sum_{t=2}^{r+2} \|X_t\|_2^2 \ge \frac{k}{64} \cdot \sum_{t=2}^{r+2} \sum_{j = 1}^{s} \phi_j(\lambda_t)^2 \ge \frac{k}{64} \cdot \sum_{j = 1}^{s} \max_{2 \le t \le r+2} \phi_j(\lambda_t)^2.
\end{equation}
    Next, we apply \Cref{cor:polynomial-growth} on $\phi_j$, which has degree at most $r$, to say that \eqref{eq:Xnorm-bound} is at least
\begin{align*}
    \frac{j}{64} \cdot \sum_{j = 1}^{s} \max_{0 \le i \le r} \phi_j\Big(\cos\Big(\frac{i}{r} \cdot \pi\Big)\Big)^2 &\ge \frac{k}{64} \cdot \sum_{j = 1}^{s} \phi_j(1+2 \eps)^2 \cdot e^{-2C \cdot \min(\sqrt{2 \eps} \cdot r, 2 \eps \cdot r^2)} \\
    &\ge \frac{j}{64} \cdot e^{-4C \cdot \sqrt{\eps} \cdot r} \cdot \sum_{j = 1}^{s} \phi_j(\lambda_1)^2.
\end{align*}
Observe, $\sum_{t=1}^{r+2} \|x_t\|_2^2 = 1$, which implies 
\begin{equation} \label{eq:square_p1_sum}
    \sum_{j=1}^s \phi_j(\lambda_1)^2 \le \frac{64}{k} \cdot e^{4C \cdot \sqrt{\eps} \cdot  r }.
\end{equation}

    Conversely, since $v_{s, 1} = \pm u_1$ for all $s$, $x_1 := \sum_{j = 1}^{s} a_{j, 1} \phi_j(\lambda_1) v_{s, 1}$ has norm at most $5 \sqrt{\log n} \cdot \sum_{j = 1}^{s} |\phi_j(\lambda_1)|$, by Triangle inequality and Part 1 of \Cref{prop:concentration}.
    Therefore,
\begin{align*}
    |\langle v, u_1 \rangle| = \|X_1\|_2 &\le 5 \sqrt{\log n} \cdot \sum_{j=1}^s |\phi_j(\lambda_1)| \\
    &\le 5 \sqrt{\log n} \cdot \sqrt{s \cdot \sum_{j=1}^s \phi_j(\lambda_1)^2} \\
    &\le 40 \sqrt{\frac{\log n \cdot s}{k}} \cdot e^{2C \cdot \sqrt{\eps} \cdot r}.
\end{align*}
    Above, the second inequality follows by Cauchy-Schwarz, and the final inequality follows by \eqref{eq:square_p1_sum}.

    Since $r = \frac{c \log n}{\sqrt{\eps}}$ for a sufficiently small constant $c$, $e^{2 C \cdot \sqrt{\eps} \cdot r} = e^{2cC \log n} \le n^{0.001}$. Also, assuming $n$ is at least a sufficiently large constant, $k \ge \frac{n}{2r}$, so
\begin{align*}
    |\langle v, u_1 \rangle| &\le 40 \cdot \sqrt{\frac{\log n \cdot s \cdot 2r }{n}} \cdot n^{0.001} \\
    &\le 80 \cdot \frac{(\log n)^{1/2}}{n^{0.499}} \cdot \frac{c \log n}{\eps^{1/2}} \\
    &\le \frac{1}{10 \cdot n^{0.498} \cdot \eps^{1/2}}.
\end{align*}
So, if $n \ge \Omega(\eps^{-2.01})$, this is at most $\sqrt{\frac{\eps}{10}}$, as desired.
\end{proof}

\begin{proof}[Proof of \Cref{thm:lb-against-block-krylov-spectral-lra}]
To see why \Cref{lem:lb-against-krylov-main} implies \Cref{thm:lb-against-block-krylov-spectral-lra}, suppose there existed $v \in \Span(\calK)$ such that $\|A(I-vv^\top)\|_{op} \le (1+\eps) \min_{\|u\|_2 = 1} \|A(I-uu^\top)\|_{op}$. Then, by \Cref{lem:alignment-to-lra}, either $\langle v, u_1 \rangle^2 > \frac{\eps}{2}$ or $\langle Av, u_1 \rangle^2 > \frac{\eps}{2}$. However, $\|A\|_{\textrm{op}} \le 1+2\eps$, which means that either $\langle v, u_1 \rangle^2 > \frac{\eps}{2}$ or $\langle \frac{Av}{\|Av\|_2}, u_1 \rangle^2 > \frac{\eps}{10}$. By increasing the number of power method iterations by $1$, we may assume both $v$ and $\frac{Av}{\|Av\|_2}$ are in $\calK$. Hence, we have obtained a contradiction with \Cref{lem:lb-against-krylov-main}, which shows that \Cref{lem:lb-against-krylov-main} implies \Cref{thm:lb-against-block-krylov-spectral-lra}.
\end{proof}

%% file: lifting-result.tex
\section{Lower bound against Arbitrary Adaptive algorithms}
\label{sec:lowerbounds-against-adaptive-algorithms}
In this section, we show that our lower bound against block Krylov algorithms extends to a lower bound against arbitrary adaptive algorithms.
The main technical result we will utilize is a general reduction theorem which shows that for a wide class of matrix-vector problems, it suffices to prove a lower bound against block Krylov algorithms. This will lead to an optimal lower bound for rank-1 spectral low-rank approximation, against an arbitrary adaptive algorithm.

\subsection{Reducing Arbitrary Adaptive Queries to Block Krylov}

Let $A = U^\top D U$, where $D$ is a diagonal matrix, $U$ is a Haar-random orthogonal matrix (see Definition~\ref{def:haar-random}), and $U$ and $D$ are independent. We consider the following model, which is a strengthening of the matrix-vector product model:

\begin{definition}[Extended Oracle Model]
\label{def:extended-oracle-model}
Given $K \in \mathbb{N}$, for all  $k \in [K]$, the algorithm chooses a new query point $v_k$, and receives the information $\{A^i v_j\}_{(i, j) \in H_k},$ where $H_k := \{(i, j): i+j \le k+1, i \ge 0, 1 \le j \le k\}$ is a set of ordered pairs of nonnegative integers. We use the following notation $\{A^i v_j\}_{S}$ for any set $S$ to denote $\{A^i v_j\}_{(i, j) \in S}$.  
\end{definition}

Note that this is clearly a stronger oracle model than the usual matrix-vector oracle, so a lower bound against algorithms in the extended oracle model implies a lower bound against algorithms in the original matrix-vector model.

\begin{definition}[Adaptive Deterministic Algorithm]
\label{def:adaptive-deterministic-algorithm}
An \textit{adaptive deterministic algorithm} $\calA$ that makes $K$ extended oracle queries (see Def~\ref{def:extended-oracle-model}) is given by a deterministic collection of functions $v_1,v_2\Paren{\cdot},\dotsc,v_{K}\Paren{\cdot}$, where $v_1$ is constant and each $v_k \Paren{\cdot}$ is a function of $\frac{k\,(k+1)}{2}-1$ inputs.
This corresponds to a sequence of queries where the $k$-th query $v_k(\{A^i v_j\}_{H_{k-1}})$ is chosen adaptively based on the information available to the algorithm at the start of iteration $k$.
(Note that $v_1$ has no inputs.)
When the choice of the inputs is clear from context, we may simply write $v_k = v_k(\{A^i v_j\}_{H_{k-1}})$ for simplicity.
\end{definition}

The main result we require, which is closely based on a recent result of~\cite{samplingLB}, is the following.

\begin{lemma}[reduction to block Krylov]\label{lem:chen_block_krylov}
    Let $v_1, v_2(\cdot), \dotsc, v_{K}(\cdot)$ be an adaptive deterministic algorithm that makes $K$ queries, where $K^2 < d$. Let $\valg_1, \valg_2, \dotsc, \valg_{K}$ be recursively defined as follows: $\valg_1 = v_1$, and $\valg_k = v_k(\{A^i \valg_j\}_{H_{k-1}})$ for $k \ge 2$. Let $z_1, \dotsc, z_{K}$ be i.i.d.\ standard Gaussian vectors. Then, from the collection $\{A^i z_j\}_{i+j \le K}$, we can construct a set of unit vectors $\vsim_1,\vsim_2, \dotsc, \vsim_{K}$, and a set of rotation matrices $\Usim_1, \Usim_2, \dotsc, \Usim_{K}$, with the following properties.
\begin{enumerate}
    \item $\vsim_k$ and $\Usim_k$ only depend on $\{A^i z_j\}_{i+j \le k}$. Moreover, $\vsim_k \in \Span(\{A^i z_j\}_{i+j \le k}).$
    \item $\big((\Usim_{1:K})^\top A (\Usim_{1:K}), \{(\Usim_{1:K})^\top A^i \vsim_j\}_{H_{K}}\big) \eqdist \big(A, \{ A^i \valg_j \}_{H_{K}}\big)$, where $\Usim_{1:K} \deq \Usim_1 \dotsm \Usim_{K}$. Here, the equivalence in distribution is over the randomness of $A$ and $\{z_i\}_{i \le K}$. 
\end{enumerate}
\end{lemma}

Importantly, Property 2 of \Cref{lem:chen_block_krylov} roughly says that the knowledge of $A^i \tilde{v}_j$ is sufficient to reconstruct the distribution of the adaptive algorithm's queries and responses.

\begin{remark}
    While the reduction lemma is written against \emph{deterministic} algorithms, it will turn out to be quite simple to remove this assumption, as the matrix $A$ has already been randomized.
\end{remark}

We defer the proof of \Cref{lem:chen_block_krylov} to \Cref{appendix:lifting}.

\subsection{General Matrix-Vector Lower Bound}

We are now ready to combine the information-theoretic lower bound against Block Krylov algorithms and our lifting statement that shows Block Krylov can simulate an Adaptive  Detereministic Algorithm (see Def~\ref{def:adaptive-deterministic-algorithm}) to obtain an information-theoretic lower bound on the number of matrix-vector products.  

\begin{theorem}[Matrix-Vector Lower Bound for Spectral LRA]
    Fix $D$ as in Definition \ref{def:hard-distribution-chebyshev-polynomial}, and suppose $0 < \eps < 0.5$. Any (potentially randomized) adaptive algorithm on the distribution $A = U^\top D U$ cannot output a unit vector $v$ such that $\langle v, u_1(A) \rangle \ge \sqrt{\frac{\eps}{10}},$ with more than $0.2$ probability over the randomness of both $A$ and the algorithm.
    Here, $u_1(\cdot)$ represents the maximum eigenvector of a symmetric matrix.
    Hence, by \Cref{lem:alignment-to-lra}, outputting a $(1+\eps)$-approximate rank-1 approximation to $A$ with more than $0.2$ probability needs $\Omega\left(\frac{\log n}{\sqrt{\eps}}\right)$ adaptive queries.
\end{theorem}

\begin{proof}
    First we assume that the algorithm is deterministic, so its behavior is characterized by functions $v_1, v_2(\cdot), \dotsc, v_K(\cdot)$, as in Lemma~\ref{lem:chen_block_krylov}.
    Using one additional query $v_{K+1}(\cdot)$, the algorithm can ensure that some linear combination of $v_{K+1}$ and $\{A^i v_j\}_{H_K}$ contains such a vector $v$ with at least $0.2$ probability. So, our goal is to show that $\Span(\{A^i \valg_j\}_{i+j \le K+1})$ cannot contain such a vector $v$ with high probability.

    By Lemma \ref{lem:chen_block_krylov}, a block Krylov algorithm that receives $\{A^i z_j\}_{i+j \le K+1}$ can generate $\tilde{v}_1, \dots, \tilde{v}_{K+1},$ where each $\tilde{v}_k \in \Span(\{A^i z_j\}_{i+j \le k})$, along with $\Usim_1, \dots, \Usim_{K+1}$, so that 
\[\big((\Usim_{1:(K+1)})^\top A (\Usim_{1:(K+1)}), \{(\Usim_{1:(K+1)})^\top A^i \vsim_j\}_{H_{K+1}}\big) \eqdist \big(A, \{ A^i \valg_j \}_{H_{K+1}}\big).\]
    Then, if with at least $0.2$ probability there exists a unit vector $v \in \Span(\{A^i \valg_j\}_{i+j \le K+1})$ with $\langle v, u_1(A) \rangle \ge \sqrt{\frac{\eps}{10}}$, then with the same probability there exists a unit vector 
\[v' \in \Span(\{U_{1:(K+1)}^\top A^i \tilde{v}_j\}_{i+j \le K+1}) = U_{1:(K+1)}^\top \cdot \Span(\{A^i z_j\}_{i+j \le K+1})\]
    with $\big\langle v', u_1\big((\Usim_{1:(K+1)})^\top A (\Usim_{1:(K+1)})\big)\big\rangle \ge \sqrt{\frac{\eps}{10}}$. However, the top eigenvector of $(\Usim_{1:(K+1)})^\top A (\Usim_{1:(K+1)})$ is $(\Usim_{1:(K+1)})^\top u_1(A)$, so this means $\langle v', (\Usim_{1:(K+1)})^\top u_1(A) \rangle \ge \sqrt{\frac{\eps}{10}}.$ In turn, this implies that the span of $\{A^i z_j\}_{i+j \le K+1}$ has a unit vector $v''$ with $\langle v'', u_1(A) \rangle \ge \sqrt{\frac{\eps}{10}}$ with at least $0.2$ probability, which contradicts \Cref{lem:lb-against-krylov-main}. Hence, no deterministic algorithm can succeed with more than $0.2$ probability.

    If the algorithm is randomized, then it uses a random seed $\xi$ that is independent of $A$. So conditional on the random seed, the algorithm will not be able to succeed with more than $0.2$ probability, which means the overall probability that the randomized algorithm successfully finds such a vector is also at most $0.2$.

    To finish, by \Cref{lem:alignment-to-lra}, a $(1+\eps)$-approximate rank-1 approximation $Avv^\top$ to $A$ requires either $\langle u_1, v \rangle^2 \ge \eps/2$ or $\langle u_1, A v \rangle^2 \ge \eps/2$. Moreover, $\langle u_1, A v \rangle^2 = \langle A u_1, v \rangle^2 = (1+2\eps)^2 \langle u_1, v \rangle^2$. Assuming $\eps \le 0.5$, this means we must have $\langle u_1, v \rangle^2 \ge \eps/10$. This concludes the proof.
\end{proof}

%% file: faster-krylov-algorithms.tex
\section{Upper Bounds for Schatten-$p$ LRA}


In this section, we show that just a direct application of Krylov methods with a single starting vector suffices for all $p =O(1)$. The query complexity will only be $ O(\eps^{-1/3} \log(1/\eps))$ and does not scale with input size.  More formally, we prove the following theorem:

\begin{theorem}[Sharper Algorithm for LRA]\label{thm:rank_1_frobenius}
    Given $A \in \mathbb{R}^{n\times d}$, $\eps>0$ and $p = O(1)$, there exists an algorithm that uses $O\Paren{ \log(1/\eps)/\eps^{1/3}}$ matrix-vector products, and outputs a unit vector $w$ such that with probability at least $0.99$, 
    \begin{equation*}
        \norm{ A\Paren{I - ww^\top} }_{\calS_p}^p \leq \Paren{1+\eps} \min_{ \norm{u}=1} \norm{A\Paren{I - uu^\top} }_{\calS_p}^p.
    \end{equation*}
\end{theorem}

It is well-known that the the optimum rank-$1$ approximation of $A$, in any Schatten-$p$ norm, is of the form $Auu^\top$, for some $u$ with norm $1$. (In fact, the optimum $u$ is precisely the top singular vector of $A$).
Hence, $A ww^\top$ is a $(1+\eps)$-approximate rank-$1$ approximation of $A$.

In order to prove Theorem \ref{thm:rank_1_frobenius}, we note the following intermediate lemmas.

\begin{lemma}[Pythagorean Theorem for Matrices] \label{prop:pythagorean_theorem}
    For any unit vector $w,$ $\|A ww^\top\|_F^2 + \|A (I-ww^\top)\|_F^2 = \|A\|_F^2$. More generally, for any projection matrix $P$, $\|A P\|_F^2 + \|A (I-P)\|_F^2 = \|A\|_F^2$.
\end{lemma}

More generally, we require the following two lemmas that generalize the Pythagorean theorem to Schatten-$p$ spaces, for any $p\geq1$.

\begin{lemma}[Generalized Pythagorean Inequality Lemma 5.5,~\cite{bakshi2022low}] \label{lem:schatten-pythagorean-initial}
    Let $A \in \BR^{n \times d}$, and $P, Q$ be projection matrices of equal rank in $\BR^{n \times n}$. Then, for any $p \ge 1$,
\[\|A\|_{\calS_p}^p \ge \|P A Q\|_{\calS_p}^p + \|(I-P) A (I-Q)\|_{\calS_p}^p.\]
\end{lemma}

\vspace{0.2in}
\begin{mdframed}

  \begin{algorithm}[Krylov Iteration (rectangular matrices)]
    \label{algo:simul_power_iter}\mbox{}
    \begin{description}
    \item[Input:] An $n \times d$ matrix $A$, target rank $k$, targest accuracy $\eps>0$, and iteration count $t = O\Paren{ p^{-1} \cdot \log(1/\eps)/\eps^{1/3} }$. 
    
    \begin{enumerate}
    \item Let $g \sim \mathcal{N}(0,I)$. Let $\calK = \left[g ; (A A^\top) g; (A A^\top)^2 g; \ldots; (A A^\top)^t g\right]$ be the resulting $d \times (t+1)$  Krylov matrix . 
    \item Compute an orthonomal basis $Q$ for the column span of $\calK$. Let $M = Q^\top A A^\top Q$. 
    \item Compute the top eigenvector vector of $M$, and denote it by $y_1$.
    \item Let $v = A^\top Q y_1/\|A^\top Q y_1\|_2$
    \end{enumerate}
    \item[Output:] $A vv^\top$.  
    \end{description}
  \end{algorithm}
\end{mdframed}

We will importantly use the following corollary of \Cref{lem:schatten-pythagorean-initial}.

\begin{lemma}[Corollary of Lemma 5.5 in~\cite{bakshi2022low}] \label{lem:schatten-pythagorean}
Given any unit vector $w$ and matrix $A \in \mathbb{R}^{n \times d}$, let $v = A^\top w/\norm{A^\top  w}_2$. Then,
\begin{equation*}
    \norm{A}_{\calS_p}^p \geq  \norm{ w w^\top  A }_{\calS_p}^p + \norm{ A\Paren{I - v v^\top} }_{\calS_p}^p
\end{equation*}
\end{lemma}
\begin{proof}
Let $P= w w^\top$ and $Q = v v^\top$. Observe that $  w w^\top A v v^\top = \frac{ w w^\top A A^\top w w^\top A    }{ \norm{A^\top w}^2  } = ww^\top A$.  
This also implies that $ww^\top A (I-vv^\top) = ww^\top A - ww^\top A vv^\top = 0$, which means that $(I-ww^\top) A (I-vv^\top) = A (I-vv^\top)$.
Invoking Lemma~\ref{lem:schatten-pythagorean-initial}, we have 
\begin{equation}
    \begin{split}
        \norm{A}_{\calS_p}^p & \geq \norm{ (ww^\top) A (vv^\top) }_{\calS_p}^p + \norm{ \Paren{I - ww^\top} A \Paren{I - vv ^\top}  }_{\calS_p}^p \\
        & = \norm{  ww^\top A }_{\calS_p}^p + \norm{  A \Paren{I - vv^\top} }_{\calS_p}^p, 
    \end{split}
\end{equation}
as desired. 
\end{proof}


We also note the following basic lemma, which will simplify our goal.

\begin{lemma}[Correlated Vectors to LRA] \label{lem:existence-implies-ub}
    Let the singular values of $A$ be $\sigma_1 \ge \sigma_2 \ge \cdots \ge \sigma_n \ge 0$.
    For $p \ge 1$, suppose that $w$ is a unit vector such that $\|A^\top w\|_2^p \ge (1+\eps) \sigma_1^p - \eps \cdot \|A\|_{\calS_p}^p$. Further, let $v = A^\top w/\norm{A^\top w}$.  Then,  $\norm{A(I-vv^\top)}_{\calS_p}^p \le (1+\eps) \cdot \min_{\norm{u}_2 = 1} \norm{A(I-uu^\top)}_{\calS_p}^p$.
\end{lemma}

\begin{proof}
    By Eckardt-Young, $\min \norm{A(I-uu^\top)}_{\calS_p}^p = \sum_{i=2}^n \sigma_i^p.$ In other words, $\min_{\norm{u}=1} \norm{A(I-uu^\top)}_{\calS_p}^p = \norm{A}_{\calS_p}^p - \sigma_1^p$. Hence, if we find a unit vector $w$ with $\norm{ww^\top A}_{\calS_p}^p \ge (1+\eps) \sigma_1^p - \eps \cdot \norm{A}_{\calS_p}^p,$ then by \Cref{lem:schatten-pythagorean}, 
\begin{align*}
    \norm{A(I-vv^\top)}_{\calS_p}^p 
    &\le \norm{A}_{\calS_p}^p - \norm{uu^\top A}_{\calS_p}^p \\
    &\le \norm{A}_{\calS_p}^p - (1+\eps) \sigma_1^p + \eps  \norm{A}_{\calS_p}^p\\
    &= (1+\eps) \cdot (\norm{A}_{\calS_p}^p - \sigma_1^p) \\
    &= (1+\eps) \cdot \min_{\|u\|_2 = 1} \norm{A(I-uu^\top)}_{\calS_p}^p,
\end{align*}
    where $v = A^\top w/\norm{A^\top w}$.
    
    Finally, since $ww^\top A$ is a rank-1 matrix and $w$ is a unit vector, its Schatten-$p$ norm is simply $ \norm{w^\top A}_2 = \norm{A^\top w}_2$. Thus, it suffices for $\norm{A^\top w}_2^p \ge (1+\eps)  \sigma_1^p - \eps \cdot \|A\|_{\calS_p}^p$.
\end{proof}

\begin{lemma}[Existence of good vectors in the Krylov Subspace]
\label{lem:existing-of-good-vectors}
    Let $\Lambda = AA^\top \in \BR^n$, with eigenvalues $\lambda_1 \ge \lambda_2 \ge \cdots \ge \lambda_n \ge 0$. Let $g \sim \mathcal{N}(0, I) \in \BR^n$ be a random Gaussian vector. Next, suppose that $t \ge C p \eps^{-1/3} \log (p/\eps)$ for some sufficiently large constant $C$, and define $\calK := \Span(g, \Lambda g, \dots, \Lambda^t g)$. Then, there exists a unit vector $w \in \calK$ such that
\[(w^\top \Lambda w)^{p/2} \ge (1+\eps) \lambda_1^{p/2} - \eps \cdot (\lambda_1^{p/2} + \cdots + \lambda_n^{p/2}) = \lambda_1^{p/2} - \eps \cdot \sum_{i=2}^n \lambda_i^{p/2}.\]
\end{lemma}

\begin{proof}
We assume WLOG that $\|A\|_{\textrm{op}} = 1$ (by scaling, note that scaling does not affect $\calK$),
so $\lambda_1 = 1$.
In addition, we may assume WLOG that $\Lambda$ is diagonal by rotating $\Lambda$: this rotates $g$ correspondingly but identity-covariance Gaussians are rotation-invariant.
We split the analysis into four cases.

\paragraph{Case 1: $\sum_{i=2}^n \lambda_i^{p/2} \ge \eps^{-1}$.} In this case, $\lambda_1^{p/2} - \eps \cdot \sum_{i=2}^n \lambda_i^{p/2} \le 0$, and since $\Lambda$ is PSD, any unit vector $w \in \calK$ satisfies $(w^\top \Lambda w)^{p/2} \ge 0 \ge \lambda_1^{p/2} - \eps \cdot \sum_{i=2}^n \lambda_i^{p/2}$.

In the remainder of the cases, given a starting vector $g = (g_1, \dots, g_d) \sim \mathcal{N}(0, I_d),$ we heavily exploit the fact that the subspace $\Span\{g, \Lambda g, \dots, \Lambda^t g\}$ contains the vector $\Paren{ \phi(\lambda_1^2) g_1, \dots, \phi(\lambda_d^2) g_d}$ for any polynomial $\phi(x)$ of degree at most $t$, since we assumed WLOG that $\Lambda$ was diagonal.

\paragraph{Case 2: $\sum_{i=2}^n \lambda_i^{p/2} < \eps^{-1}$, and the number of eigenvalues $\lambda_i$ in the range $[1-\frac{1}{2p}, 1]$ is at least $\eps^{-1/3}+1$.} In this case, first note that $\sum_{i=2}^n \lambda_i^{p/2} \ge (1-\frac{1}{2p})^{p/2} \cdot \eps^{-1/3} \ge \frac{\eps^{-1/3}}{2}$, since there are at least $\eps^{-1/3}$ eigenvalues of $\Lambda$, not including $\lambda_1 = 1$, that are at least $1 - \frac{1}{2p}$. We first consider the shifted/scaled Chebyshev polynomial $T_{t, \eps^{2/3}/(2p)}$ from \Cref{cor:magic-polynomial}, where $t = O(p^{1/2}/\eps^{1/3} \cdot \log (p/\eps))$. Note that $T_{t, \eps^{2/3}/(2p)}(1) = 1$ and 
\[|T_{t, \eps^{2/3}/(2p)}(x)| \le e^{-c \cdot \min(t^2 \cdot \eps^{2/3}/2p, t \cdot \eps^{1/3}/\sqrt{2p})} \le \frac{\eps^2}{p}\]
for all $x \in [0, 1-\frac{\eps^{2/3}}{2p}].$ \ainesh{why is it possible to have such a chebyshev polynomial? not obvious.} By letting $\hat{\phi}(x) = x^{\lceil p/2 \rceil} \cdot T_{t, \eps^{2/3}/(2p)}(x)$, we have that $\hat{\phi}(1) = 1$ and $|\hat{\phi}(x)| \le x^{p/2} \cdot \frac{\eps^2}{p}$ for all $x \in [0, 1-\frac{\eps^{2/3}}{2p}].$ For $\hat{t} \ge C \cdot p \cdot \eps^{-1/3} \log (1/\eps)$ for a sufficiently large constant $C$, $\hat{\phi}(x)$ has degree at most $\hat{t}$.

Now, with probability at least $0.99$ over the randomness of $g \sim \mathcal{N}(0, I)$, $|g_1| \ge 0.01$, and since $\hat{\phi}(\lambda_1) = \hat{\phi}(1) = 1$, we have 
$|\hat{p}(\lambda_1) g_1| \ge 0.01$.

 Next, for any $i$ such that $\lambda_i \le 1-\frac{\eps^{2/3}}{2p},$ $|\hat{\phi}(\lambda_i)| \le \frac{\eps^2}{p} \cdot \lambda_i^{p/2}$. Therefore, for each such $i$, 
\begin{equation*}
    \expecf{g \sim \mathcal{N}(0, I)}{(\hat{\phi}(\lambda_i) g_i)^2} \leq \frac{\eps^4}{p^2} \cdot \lambda_i^{p}. 
\end{equation*}
Now, since $\lambda_i \le 1$ for all $i$, and by our assumption, we have $\sum_{i=2}^n \lambda_i^{p} \le \sum_{i=2}^n \lambda_i^{p/2} < \eps^{-1}.$ Therefore,
\begin{equation}
\label{eqn:bounding-expectation-phi-hat-lambda}
\begin{split}
    \expecf{{g \sim \mathcal{N}(0, I)} }{\sum_{i: \lambda_i  \le 1-(\eps^{2/3}/(2p))} (\hat{\phi}(\lambda_i) g_i)^2 } 
   =   \sum_{i: \lambda_i \le 1-(\eps^{2/3}/(2p))} \hat{\phi}(\lambda_i)^2
    & \le \frac{\eps^4}{p^2} \cdot \sum_{i: \lambda_i \le 1-(\eps^{2/3}/(2p))} \lambda_i^p \\
    & \le \frac{\eps^4}{p^2} \cdot \sum_{i=2}^n \lambda_i^{p/2} \\
    &\le \frac{\eps^3}{p^2}.
\end{split}
\end{equation}

Therefore, by Markov's inequality, with probability at least $0.99$ over $g$, $\sum_{i: \lambda_i \le 1-(\eps^{2/3}/(2p))} (\hat{\phi}(\lambda_i) g_i)^2 \le 100 \eps^{3}/p^2$. Hence, with probability at least 0.98 over $g$, 
\begin{equation}
\label{eqn:bounding-sum-lambda_i}
    \sum_{i: \lambda_i \le 1-\eps^{2/3}/(2p)} (\hat{\phi}(\lambda_i) g_i)^2 \le \frac{10^6 \cdot \eps^3}{p^2} \cdot (\hat{\phi}(\lambda_1) g_1)^2.
\end{equation}

Next, we define $\phi(x)$ to be a scaled version of $\hat{\phi}(x)$ (i.e., $\phi(x) = \gamma \cdot \hat{\phi}(x)$ for some parameter $\gamma$), so that the vector $(\phi(\lambda_i) g_i)_{i=1}^n$ has unit norm. It follows from equation \eqref{eqn:bounding-sum-lambda_i}, that with probability at least $0.98$,
\begin{equation*}
    \sum_{i: \lambda_i \le 1-\eps^{2/3}/(2p)} (\phi(\lambda_i) g_i)^2 \le \frac{10^6 \cdot \eps^3}{p^2} \cdot (\phi(\lambda_1) g_1)^2 \le \frac{10^6 \cdot \eps^3}{p^2},
\end{equation*}
where the first inequality above holds by \eqref{eqn:bounding-sum-lambda_i} and the second inequality holds because $(\phi(\lambda_i) g_i)_{i=1}^n$ has unit norm so $|\phi(\lambda_1) g_1| \le 1$.
\ainesh{i am confused by this, (11) is scale invariant, but only $\hat{\phi}(1) =1$, $\phi(1)= \gamma$...}
This implies that $\sum_{i: \lambda_i > 1-\eps^{2/3}/2p} (\phi(\lambda_i) g_i)^2 \ge 1-10^6 \eps^3/p^2$. So, for the normalized vector $w = (\phi(\lambda_i) g_i)_{i=1}^n,$ we have
\begin{align*}
    w^\top \Lambda w &= \sum_{i=1}^n (\phi(\lambda_i) g_i)^2 \cdot \lambda_i \\
    &\ge \sum_{i: \lambda_i > 1-\eps^{2/3}/2p} (\phi(\lambda_i) g_i)^2 \cdot \left(1-\frac{\eps^{2/3}}{2p}\right) \\
    &\ge \left(1-\frac{10^6 \eps^3}{p^2}\right) \cdot \left(1-\frac{\eps^{2/3}}{2p}\right) \\
    &\ge 1-\frac{\eps^{2/3}}{1.5 p},
\end{align*}
    where the last line assumes that $\eps$ is smaller than some small but fixed constant.

    Therefore, $(w^\top \Lambda w)^{p/2} \ge (1-\eps^{2/3}/(1.5 p))^{p/2} \ge 1-\frac{\eps^{2/3}}{2}$. However, we know that $\sum_{i=2}^n \lambda_i^{p/2} \ge \frac{\eps^{-1/3}}{2}$ which means that
\[\lambda_1^{p/2} - \eps \cdot \sum_{i=2}^n \lambda_i^{p/2} \le 1 - \frac{\eps^{2/3}}{2} \le (w^\top \Lambda w)^{p/2},\]
    as desired.

%

\textbf{Case 3: $\frac{1}{2} \le \sum_{i=2}^n \lambda_i^{p/2} \le \eps^{-1}$, and the number of eigenvalues $\lambda_i$ in the range $[1-\frac{1}{2p}, 1]$ is at most $\eps^{-1/3}+1$.} In this case, we consider the polynomial $$\hat{\phi}(x) = 
x^{t/2} \cdot \prod_{i: 1-\eps/2p \ge \lambda_i \ge 1-1/2p} (x-\lambda_i).$$ 

Assuming that $t \ge C \cdot p \cdot \eps^{-1/3} \log (p/\eps)$, $\hat{\phi}(x)$ has degree at most $t$. Note that 
\begin{equation} \label{eq:phi1-bound}
    \hat{\phi}(1) = \prod_{i:1-\eps/2p \ge \lambda_i \ge 1-1/2p} (1-\lambda_i) \ge \left(\frac{\eps}{2p}\right)^{\eps^{-1/3}+1} \ge 2^{-O(\eps^{-1/3} \log (p/\eps))},
\end{equation}
since every $1-\lambda_i$ term in the product is at least $\eps/(2p)$ and there are at most $\eps^{-1}+1$ such terms. In addition, $\hat{\phi}(\lambda_i) = 0$ for all $\lambda_i \in [1-\frac{1}{2p}, 1-\frac{\eps}{2p}],$ since one of the $x-\lambda_i$ terms vanishes.
Finally, for all $\lambda_i \in [0, 1-\frac{1}{2p}],$ we know that $\prod_{i:1-\eps/2p \ge \lambda_i \ge 1-1/2p} (x-\lambda_i)$ has magnitude at most $1$, and for $t \ge C \cdot p \cdot \eps^{-1/3} \log (p/\eps)$ for a sufficiently large constant $C$, 
\begin{equation*}
    |\lambda_i^{(t-p)/2}| \le (1-1/2p)^{(t-p)/2} \le 2^{-C \eps^{-1/3} \log(p/\eps) /10} \le \hat{\phi}(1) \cdot \frac{\eps^2}{p}.
\end{equation*}
Above, the first inequality holds because we are only considering $\lambda_i \le 1-\frac{1}{2p}$, the second inequality holds by our assumption that $t \ge C \cdot p \cdot \eps^{-1/3} \log (p/\eps)$, and the third inequality holds by \eqref{eq:phi1-bound}
So, $|\hat{\phi}(\lambda_i)| \le \hat{\phi}(1) \cdot \lambda_i^{p/2} \cdot \frac{\eps^2}{p}$ for all $i$ with $\lambda_i \le 1-\eps/2p$.
\ainesh{can we add some exposition for why these inequalities are true? }

With at least $0.99$ probability over $g \sim \mathcal{N}(0, I)$, $|\hat{\phi}(1) \cdot g_1| \ge 0.01 \cdot \hat{\phi}(1)$. Also,
\begin{equation}
    \begin{split}
        \expecf{g \sim \mathcal{N}(0, I)}{ \sum_{i:\lambda_i \le 1-\eps/2p} (\hat{\phi}(\lambda_i) g_i)^2} & = \sum_{i: \lambda_i \le 1-\eps/2p} \hat{\phi}(\lambda_i)^p \\
        & \le \hat{\phi}(1)^2 \cdot \frac{\eps^4}{p^2} \cdot \sum_{i: \lambda_i \le 1-\eps/2p} \lambda_i^p \\
        & \le \hat{\phi}(1)^2 \cdot \frac{\eps^4}{p^2} \cdot \sum_{i = 2}^n \lambda_i^{p/2} \\
        & \le \hat{\phi}(1)^2 \cdot \frac{\eps^3}{p^2}.
    \end{split}
\end{equation}

Therefore, by Markov's inequality, with probability at least $0.99$ over $g$, $\sum_{i: \lambda_i \le 1-\eps} (\hat{\phi}(\lambda_i) g_i)^2 \le 100 \frac{\eps^3}{p^2} \cdot \hat{\phi}(1)^2$. Hence, with probability at least 0.98 over $g$, $\sum_{i: \lambda_i \le 1-\eps} (\hat{P}(\lambda_i) g_i)^2 \le 10^6 \cdot \frac{\eps^3}{p^2} \cdot (\hat{P}(\lambda_1) g_1)^2$.

As in the second case, we let $\phi(x)$ be a scaled version of $\hat{\phi}(x)$, so that $(\phi(\lambda_i) g_i)_{i=1}^n$ has unit norm. Then, $\sum_{i: \lambda_i \le 1-\eps/2p} (\phi(\lambda_i) g_i)^2 \le 10^6 \cdot \frac{\eps^3}{p^2},$ so $\sum_{i: \lambda_i > 1-\eps/2p} (\phi(\lambda_i) g_i)^2 \ge 1-10^6 \frac{\eps^3}{p^2}$. Therefore,
\begin{align*}
    w^\top \Lambda w
    &= \sum_{i=1}^n (\phi(\lambda_i) g_i)^2 \cdot \lambda_i \\
    &\ge \sum_{i: \lambda_i > 1-\eps/2p} (\phi(\lambda_i) g_i)^2 \cdot \left(1 - \frac{\eps}{2p}\right) \\
    &\ge \left(1 - 10^6 \cdot \frac{\eps^3}{p^2}\right) \cdot \left(1-\frac{\eps}{2p}\right) \\
    &\ge 1-\frac{\eps}{1.5 p}
\end{align*}

    Therefore, $(w^\top \Lambda w)^{p/2} \ge (1-\eps/(1.5 p))^{p/2} \ge 1-\frac{\eps}{2}$. However, we know that $\sum_{i=2}^n \lambda_i^{p/2} \ge \frac{1}{2},$ which means that
\[\lambda_1^{p/2} - \eps \cdot \sum_{i=2}^n \lambda_i^{p/2} \le 1 - \frac{\eps}{2} \le (w^\top \Lambda w)^{p/2},\]
    as desired.

\paragraph{Case 4: $\sum_{i=2}^n \lambda_i^{p/2} < \frac{1}{2}.$} In this case, $\lambda_i \le 1 - \frac{1}{2p}$ for all $i \ge 2$. Therefore, we can set $\hat{\phi}(x) = x^L$ for some $L = O(\log (p/\eps))$, to obtain $\hat{\phi}(1) = 1$ and $\hat{\phi}(\lambda_i) \le \frac{\eps^2}{p} \cdot \lambda_i^{p/2}$ for all $i \ge 2$. With probability at least $0.99$ over $g \sim \mathcal{N}(0, I),$ $|\hat{\phi}(1) \cdot g_1| \ge 0.01$, since $\hat{P}(1) = 1$, and 
\begin{equation*}
    \expecf{g \sim \mathcal{N}(0, I)}{\sum_{i=2}^n (\hat{\phi}(\lambda_i) g_i)^2 } = \sum_{i=2}^n \hat{P}(\lambda_i)^2 \le \sum_{i=2}^n \frac{\eps^4}{p^2} \cdot \lambda_i^p \le \frac{\eps^4}{p^2} \cdot \left(\sum_{i=2}^n \lambda_i^{p/2}\right).
\end{equation*}
Define $\tau = \sum_{i=2}^n \lambda_i^{p/2}$.
Again, by Markov's inequality, $\sum_{i=2}^n (\hat{\phi}(\lambda_i) g_i)^2 \le 100\frac{\eps^4}{p^2} \cdot \tau$ with probability at least $0.99$ over $g$. Thus, with probability at least $0.98$ over $g$, $\sum_{i=2}^n (\hat{\phi}(\lambda_i) g_i)^2 \le 10^6 \cdot \frac{\eps^4}{p^2} \cdot \tau \cdot (\hat{\phi}(\lambda_1) g_1)^2$.

As in the second and third cases, we let $\phi(x)$ be a scaled version of $\hat{\phi}(x)$, so that $w := (\phi(\lambda_i) g_i)_{i=1}^n$ has unit norm. We have that $\sum_{i=2}^n (\phi(\lambda_i) g_i)^2 \le 10^6 \cdot \frac{\eps^4}{p^2} \cdot \tau,$ which means $(\phi(\lambda_1) g_1)^2 \ge 1 - O(\frac{\eps^4}{p^2} \cdot \tau)$. Hence, $w^\top \Lambda w \ge \lambda_1 \cdot (\phi(\lambda_1) g_1)^2 \ge 1-O(\frac{\eps^4}{p^2} \cdot \tau)$, which means that $(w^\top \Lambda w)^{p/2} \ge 1-O(\frac{\eps^4}{p} \cdot \tau)$. However, $\lambda_1^{p/2} - \eps \cdot \sum_{i=2}^n \lambda_i^{p/2} = 1 - \eps \cdot \tau$. Thus, 
\[(w^\top \Lambda w)^{p/2} \ge \lambda_1^{p/2} - \eps \cdot \sum_{i=2}^n \lambda_i^{p/2},\]
which completes the proof.
\end{proof}

We are now ready to complete the proof of Theorem~\ref{thm:rank_1_frobenius}.

\begin{proof}[Proof of Theorem~\ref{thm:rank_1_frobenius}]
Recall, $t = O\Paren{ p^{-1} \log(1/\eps)/\eps^{1/3}}$. 
In $O(t)$ matrix-vector computations, we can compute $\calK$ and compute $\max_{w \in \calK: \|w\|_2=1} \|A^\top w\|_2$. This is equivalent to $w = Q y_1$, where $y_1$ is the top singular vector of $A^\top Q$ for $Q$ an orthonormal basis for the column span of $\calK$, or equivalently, the top eigenvector of $Q^\top AA^\top Q$. If there exists a unit vector in the subspace $\calK := \Span(g, (AA^\top) g, \dots, (AA^\top)^t g)$ satisfying the assumption of \Cref{lem:existence-implies-ub}, then $w$ will also satisfy the assumption, since $w$ is defined  to maximize $\|A^\top w\|_2$ over $w \in \calK$. Thus, by \Cref{lem:existence-implies-ub}, we have that for $v = A^\top w/\norm{A^\top w},$ $Avv^\top$ is a $(1+\eps)$-approximate rank-$1$ approximation in Schatten-$p$ norm. Since $\|A^\top w\|_2 = \sqrt{w^\top AA^\top w}$, \Cref{lem:existing-of-good-vectors} implies the existence of such a vector, concluding the proof.  
%
 %
\end{proof}

%% file: biblio.tex
\newcommand{\etalchar}[1]{$^{#1}$}

%% file: lifting-proof.tex
\section{Proof of \Cref{lem:chen_block_krylov}} \label{appendix:lifting}

In this section, we prove \Cref{lem:chen_block_krylov}. The proof will essentially copy that of~\cite{samplingLB} (which we have received explicit permission from the authors to do), apart from a few details to provide a minor generalization.

\subsection{Additional Preliminaries}

We recall the definitions of the Extended Oracle Model (\Cref{def:extended-oracle-model}) and adaptive deterministic algorithms (\Cref{def:adaptive-deterministic-algorithm}).

First, we note that, in the extended oracle model, we can assume that each $v_k$ is a unit vector orthogonal to its inputs.
\begin{lemma}[Extended Oracle and Orthogonal Queries]
\label{lem:extended-oracle-queries-are-orthogonal}
For $k \in [2, K]$, let $v_k$ be as stated in Definition~\ref{def:adaptive-deterministic-algorithm} and let $\{A^{i} v_j\}_{H_{k-1}}$ be as stated in Definition~\ref{def:extended-oracle-model}. Then, we may assume WLOG that $v_k$ is orthogonal to the subspace spanned by the vectors in $\{A^{i} v_j\}_{H_{k-1}}$. 
\end{lemma}
\begin{proof}
Assume for sake of contradiction that this were not the case. Then, we can decompose $v_k = \sum_{(i, j) \in H_{k-1}} c_{i,j} A^i v_j + c^\perp v_k^\perp$ where $v_k^\perp$ is a unit vector orthogonal to $\{A^i v_j\}_{H_{k-1}}$ and each $c_{i, j}$ and $c^\perp$ is a scalar.
At the end of iteration $k$, the new information obtained by the algorithm is $\{A^i v_j\}_{i+j=k+1, j \le k}$.
For all $(i,j) \ne (1,k)$, the new information does not depend on $v_k$.
Also, $A v_k = \sum_{(i, j) \in H_{k-1}} c_{i,j} A^{i+1} v_j + c^\perp A v_k^\perp$, where each $A^{i+1} v_j$ is information obtained by the algorithm at the end of iteration $k+1$ regardless (due to our extended query model).
Since $(i+1, j) \in H_k$ if $(i, j) \in H_{k-1}$, and since $(1, k) \in H_k$, this expression shows that the algorithm would receive the same amount of information (or more, if $c^\perp = 0$) if it queries $v_k^\perp$ instead of $v_k$.
Applying this reasoning inductively proves the claim.
\end{proof}


We compare to a Block Krylov algorithm that makes i.i.d.\ standard Gaussian queries $z_1,z_2,\dotsc,z_K$ and then receives $\{A^i z_j\}$ for all $i, j \le K$. Recall, the Block Krylov algorithm does not make \emph{adaptive} queries, it is easier to prove lower bounds against Block Krylov algorithms. Our goal is to now show that Block Krylov algorithms can simulate an adaptive deterministic algorithm.

\subsection{Conditioning lemma}\label{scn:conditioning}

We start by proving a general conditioning lemma which will be invoked repeatedly in the reduction to Block Krylov algorithms.
We implicitly assume that all mappings are measurable, in order to avoid undue technical issues. This lemma roughly shows that if the adaptive algorithm knows $\{A^i v_j\}_{H_k},$ the posterior distribution of $A$ given $\{A^i v_j\}_{H_k}$ is indeed rotationally symmetric on the orthogonal complement $\{A^i v_j\}_{H_k}$.

We will use the notation $\eqdist$ to denote that two random variables are equal in probability distribution (possibly conditioned on other information).


\begin{lemma}[conditioning lemma, preliminary version] \label{lem:random_rotation_fixed_subspace}
    Let $U$ be a Haar-random orthogonal matrix, and $A = U^\top D U$, where $D$ is a fixed positive diagonal matrix. Suppose that $\calA$ is an adaptive deterministic algorithm that generates extended oracle queries $v_1, \dots, v_{K}$, and after the $k$th query knows $A^i v_j$ for all $(i, j) \in H_k$.    
    For any integer $m \ge 1$, let $k$ be the integer such that $\frac{k(k+1)}{2} \le m < \frac{(k+1)(k+2)}{2},$ i.e., $m$ is at least the $k$th triangular number but less than the $(k+1)$th triangular number.
    Consider the order of vectors $v_1, A v_1, v_2, A^2 v_1, A v_2, v_3, A^3 v_1, \dots$ (this enumerates $A^i v_j$ in order of $i+j$, breaking ties with smaller values of $j$ first).
    Let $W_m$ be the set of first $m$ of these vectors and $X_k$ be the set $\{v_1, \dots, v_k\}$. Let $V$ be a Haar-random orthogonal matrix fixing $W_m$ and acting on the orthogonal complement $W_m^\perp$. Then, $(X_k, U) \eqdist (X_k, U V)$.
\end{lemma}

Before proving this lemma, we note that since the algorithm is deterministic and $D$ is fixed, $W_m$ and $X_k$ are deterministic functions of $A$, and thus of $U$. Hence, we can write $v_k(U'), W_m(U'), X_k(U')$ to be the $v_k, W_m, X_k$ that would have been generated if we started with $A' = (U')^\top D U'$. (If no argument is given, $v_k, W_m, X_k$ are assumed to mean $v_k(U), W_m(U), X_k(U)$, respectively.) We note the following proposition.

\begin{proposition}[fixing the first $m$ queries and responses]
\label{prop:invariant}
    Suppose that $V$ is any orthogonal matrix fixing $W_m(U)$. Then, $W_m(U) = W_m(UV)$.
\end{proposition}

\begin{proof}
    We prove $W_{m'}(U) = W_{m'}(UV)$ for all $m' \le m$. The base case of $k = 1$ is trivial, since $v_1$ is fixed. We now prove the induction step for $m'$.
    
    If $m' \le m$ is a triangular number, $m' = \frac{k(k+1)}{2}$, then the $m'$th vector in $W_m$ is $v_k$. But note that $v_k(U)$ is a deterministic function of $W_{m'-1}(U)$, and $v_k(UV)$ is the same deterministic function of $W_{m'-1}(UV)$. Hence, if the induction hypothesis holds for $m'-1$, it also holds for $m$.

    If $m' \le m$ is not a triangular number, then the $m'$th number in $W_m(U)$ is $A^i v_j$ for some $i \ge 1$. Likewise, the $m'$th number in $W_m(UV)$ is $V^\top A^i V \cdot v_j(UV)$. Since $i \ge 1$, we know that $v_j(U) = v_j(UV)$, by the induction hypothesis on $\frac{j(j+1)}{2} < m'$. But, we know that $V$ fixes $W_m$, which means it fixes $v_j$ and $A^i v_j$. Thus, $V^\top A^i V v_j(UV) = V^\top A^i V v_j = A^i v_j$.
\end{proof}

We are now ready to prove Lemma \ref{lem:random_rotation_fixed_subspace}.

\begin{proof}[Proof of Lemma~\ref{lem:random_rotation_fixed_subspace}]
    We prove this by induction on $m$. For the base case $m = 1$, $U$ is a random matrix and $V$ is a random matrix that fixes $v_1$. Note that $v_1$ is chosen independently of $A$ (and thus of $U$), so $U$ and $V$ are independent. Even for any fixed $V$, the distribution $U V$ is a uniformly random orthogonal matrix, so overall $U \eqdist U V$. Also, $v_1$ is deterministic, so $(v_1, U) \eqdist (v_1, U V)$.
    
    For the induction step, we split the proof into $2$ cases. The proofs in both cases will be very similar, but with minor differences.
    
    \textbf{Case 1: $m$ is a triangular number.} This means that the $m$th vector added is $v_k$, where $m = \frac{k(k+1)}{2}$. Let $V_1$ be a random orthogonal matrix fixing $W_{m-1}$ and $V_2$ be a random orthogonal matrix fixing $W_m$. Our goal is then to show $(X_k, U) \eqdist (X_k, U V_2)$.

    To make this rigorous, we note an order of generating the random variables. First, we generate $U$ randomly: $W_m$ and $X_k$ are deterministic in terms of $U$. Next, we define $V_1$ to be a random rotation fixing $W_{m-1}$. Finally, we define $V_2$ to be a random rotation fixing $W_m$, where $V_1, V_2$ are conditionally independent on $U$.

    First, we prove that $(X_k, U) \eqdist (X_k, UV_1)$. Note that $U \eqdist UV_1$ by our inductive hypothesis. In addition, since $V_1$ fixes $W_{m-1}(U)$, $W_{m-1}(U) = W_{m-1}(UV_1)$ by \Cref{prop:invariant}. Since $m = \frac{k(k+1)}{2}$ is a triangular number, $X_k(\cdot)$ is a deterministic function of $W_{m-1}(\cdot)$, which means $X_k(U) = X_k(UV_1)$. Hence, $(X_k, U) \eqdist (X_k(UV_1), UV_1) = (X_k, UV_1)$.
        
    Next, we prove that $(X_k, U V_2) \eqdist (X_k, U V_1 V_2)$. It suffices to prove that 
\[(X_k, U, V_2) \eqdist (X_k, UV_1, V_2).\]
    To do so, we first show that $V_2 = f(U, R)$, where $f$ is a deterministic function and $R$ represents a random orthogonal matrix over $d-\dim(W_m)$ dimensions that is independent of $U$. (Recall that $W_m$ is a deterministic function of $U$.) To define $f(U, R)$, we consider some deterministic map that sends each $W_m$ to a set of $d-\dim(W_m)$ basis vectors in $W_m^\perp$. We then define $V_2 = f(U, R)$ to act on $W_m^\perp$ using $R$ and the correspondence of basis vectors. Since $W_m$ and $X_k$ are deterministic in terms of $U$, this means $f(U, R)$ is well-defined.
    We will now show that 
\[V_2 = f(U, R) = f(UV_1, R) \hspace{0.5cm} \text{and} \hspace{0.5cm} X_k = X_k(UV_1).\]
    Since $U \eqdist UV_1$ by our inductive hypothesis, 
\[(X_k, U, V_2) \eqdist (X_k(UV_1), UV_1, f(UV_1, R)) = (X_k, UV_1, V_2).\]
    By Proposition \ref{prop:invariant}, $W_{m-1}(U) = W_{m-1}(UV_1),$ and since $X_k(\cdot)$ is deterministic given $W_{m-1}(\cdot)$ for $m = \frac{k(k+1)}{2}$, $X_k(U) = X_k(UV_1)$. This implies $W_m(U) = W_m(UV_1),$ which means $f(UV_1, R) = f(U, R)$, since $f(\cdot, R)$ only depends on $W_m(\cdot)$ and $R$. This completes the proof.
    
    
    Next, we show that $(X_k, U V_1 V_2) \eqdist (X_k, U V_1)$. Since we chose the order with $U$ being defined first, we are allowed to condition on $U$. Since $X_k$ is deterministic in terms of $U$, it suffices to show that $V_1 V_2|U \eqdist V_1|U$. Since $W_{m-1}, W_m$ are also deterministic given $U$, note that $V_1$ is a uniformly random orthogonal matrix fixing $W_{m-1},$ and $V_2$ is a random orthogonal matrix fixing $W_{m} \supset W_{m-1}$. Since $V_1$ and $V_2$ are conditionally independent given $U$, this means $V_1 V_2|U$ is a uniformly random orthogonal matrix fixing $W_{m-1}$, so $V_1 V_2|U \eqdist V_1|U$.
    
    In summary, we have that
\begin{align*}
    (X_k, U) &\eqdist (X_k, U V_1) \\
    &\eqdist (X_k, U V_1 V_2) \\
    &\eqdist (X_k, U V_2).
\end{align*}
    
    \paragraph{Case 2: $m$ is not a triangular number.}
    Again, let $V_1$ be a random orthogonal matrix fixing $W_{m-1}$ and $V_2$ be a random orthogonal matrix fixing $W_m$. Our goal is again to show that $(X_k, U) \eqdist (X_k, U V_2)$. 

    First, we again have $(X_k, U V_1) \eqdist (X_k, U)$ by our inductive hypothesis.
    
    Next, we show that $(X_k, U V_2) \eqdist (X_k, U V_2 V_1)$. It suffices to prove that 
\[(X_k, U, V_2) \eqdist (X_k, UV_1, V_1^\top V_2 V_1),\]
    since $(U V_1) (V_1^\top V_2 V_1) = U V_2 V_1$.
    We recall the random variable $R$ and use the same function $V_2 = f(U, R)$.
    Since we have already shown that $U \eqdist UV_1$, this implies that $(X_k, U, V_2) \eqdist (X_k(UV_1), UV_1, f(UV_1, R))$.
    Since $m$ is not triangular, $X_k(\cdot)$ is contained in $W_{m-1}(\cdot)$, so by Proposition \ref{prop:invariant}, $X_k(U) = X_k(UV_1)$. So, we have 
\[(X_k, U, V_2) \eqdist (X_k(UV_1), UV_1, f(UV_1, R)) = (X_k, UV_1, f(UV_1, R)).\]
    Now, if we fix $U$ and $V_1$, $W_{m-1}(UV_1) = W_{m-1}(U)$ by \Cref{prop:invariant}. However, since the $m$th $(i, j)$ pair has $i \ge 1$ when $m$ is not triangular, the final vector in $W_m(UV_1)$ will be $V_1^\top A^i V_1 v_j = V_1^\top (A^i v_j)$. For fixed $U, V_1$, $f(U, R)$ is a random rotation fixing $W_{m-1}$ and $A^i v_j$, but $f(UV_1, R)$ is a random rotation fixing $W_{m-1}$ and $V_1^\top (A^i v_j)$. Since $V_1^\top$ fixes $W_{m-1}$ by how we defined $V_1$, this means that for fixed $U, V_1$, $f(U, R)$ is a random rotation fixing $W_m$ but $f(UV_1, R)$ is a random rotation fixing $V_1^\top W_m$. Therefore, conditioned on $U, V_1$, $f(UV_1, R)$ has the same distribution as $V_1^\top f(U, R) V_1$. Since $X_k$ is deterministic in terms of $U$, this means
\[(X_k, UV_1, f(UV_1, R))|U, V_1 \eqdist (X_k, UV_1, V_1^\top f(U, R) V_1)|U, V_1.\]
    We can remove the conditioning to establish that $(X_k, UV_1, f(UV_1, R)) \eqdist (X_k, UV_1, V_1^\top f(U, R) V_1) = (X_k, UV_1, V_1^\top V_2 V_1),$ which completes the proof.
        
    Next, we show that $(X_k, U V_2 V_1) \eqdist (X_k, U V_1)$. The proof is essentially the same as in the case when $m$ is triangular. We again condition on $U$, and we have that $V_2 V_1|U \eqdist V_1|U$ have the same distribution as uniform orthogonal matrices fixing $W_{m-1}(U)$. Since $X_k$ is a deterministic function of $U$, this means $(X_k, U V_2 V_1)|U \eqdist (X_k, U V_1)|U,$ and removing the conditioning finishes the proof.
    
    In summary, 
\begin{align*}
    (X_k, U) &\eqdist (X_k, U V_1) \\
    &\eqdist (X_k, U V_2 V_1) \\
    &\eqdist (X_k, U V_2). \qedhere
\end{align*}
\end{proof}

We now prove our main conditioning lemma, which will be a modification of Lemma \ref{lem:random_rotation_fixed_subspace}.

\begin{lemma}[Conditioning Lemma] \label{lem:conditioning_lemma}
    Let all notation be as in Lemma \ref{lem:random_rotation_fixed_subspace}, and let $V_0$ be a fixed orthogonal matrix fixing $W_m$. Importantly, $V_0$ is a deterministic function only depending on $W_m$ (and not directly on $U$). Then, $(X_k, U) \eqdist (X_k, UV_0)$.
\end{lemma}

\begin{proof}
    First, note that since $V_0$ is a deterministic function of $W_m$, it is also a deterministic function of $U$. We can write $V_0(\cdot)$ as this function, and $V_0 = V_0(U)$.

    Now, Lemma \ref{lem:random_rotation_fixed_subspace} proves that $(X_k, U) \eqdist (X_k, UV)$. Note that conditioned on $U$, $V$ is a random matrix fixing $W_m$ and $V_0$ is a fixed matrix fixing $W_m$, which means that $VV_0|U \eqdist V|U$. Hence, $(X_k, UV) \eqdist (X_k, UVV_0)$. But from Proposition \ref{prop:invariant}, $X_k(UV) = X_k(U)$ and $W_m(UV) = W_m(U)$, which means that $V_0(\cdot)$, which only depends on $W_m(\cdot)$, satisfies $V_0(UV) = V_0(U)$. Hence, because $U \eqdist UV$, we have $(X_k, UVV_0) = (X_k(UV), UV \cdot V_0(UV)) \eqdist (X_k(U), U \cdot V_0(U)) = (X_k, UV_0)$.

    In summary, we have that $(X_k, U) \eqdist (X_k, UV) \eqdist (X_k, UVV_0) \eqdist (X_k, UV_0)$, which completes the proof.
\end{proof}

\subsection{From Query Algorithms to Block Krylov Algorithms}\label{scn:from_query}


We now aim to prove \Cref{lem:chen_block_krylov}, which implies that any adaptive deterministic algorithm in the extended oracle model can be simulated by rotating the output of a block Krylov algorithm.


First, we describe how to construct $\vsim_k$. Let $\vsim_1 = \frac{z_1}{\norm{z_1}}$, and for $k \ge 2$, let $\vsim_k$ be the unit vector parallel to the component of $z_k$ that is orthogonal to the span of $\{A^i z_j\}_{H_{k-1}}$. (With probability $1$, this is well-defined.) Equivalently, we can let $Q_k$ be an orthogonal basis for $\Span\{A^i z_j\}_{H_{k-1}}$, and define $\vsim_k := \Paren{I- Q_k Q_k^\top} z_k  / \norm{ \Paren{I- Q_k Q_k^\top} z_k}_2 $.

We note that $\tilde{v}_k$ can be written in terms of of $A^i z_j$, and likewise, $\tilde{z}_k$ can be written in terms of $A^i \tilde{v}_j$. Formally, we have the following.
\begin{proposition}
    For every $k \ge 1$, $\tilde{v}_k$ is a linear combination of $\{A^i z_j\}_{i+j \le k}$, and $z_k$ is a linear combination of $\{A^i \tilde{v}_j\}_{i+j \le k}$.
\end{proposition}

\begin{proof}
    By definition, $\vsim_k$ is a linear combination of $\{A^i z_j\}_{H_{k-1}}$ and $z_k$, so it is a linear combination of $\{A^i z_j\}_{i+j \le k}$. Therefore, we can construct the set $\{A^i \vsim_j\}_{H_k}$ as a linear combination of the set $\{A^i z_j\}_{H_k}$, for all $k \le K$.
    
    We now show that $z_k$ is a linear combination of $\{A^i \tilde{v}_j\}_{i +j \le k}$, by induction. The base case of $k = 1$ is trivial. Now, assume the induction hypothesis for $k-1$. With probability $1$, $\tilde{v}_k$ is nonzero, so $z_{k}$ is a linear combination of $\tilde{v}_{k}$ and $\{A^i z_j\}_{H_{k-1}}$. By the inductive hypothesis, each $z_j$ is a linear combination of $\{A^{i'} z_{j'}\}_{i'+j' \le j}$ for $j < k$, which means for $(i, j) \in H_{k-1}$, $A^i z_j$ can be written as a linear combination of $A^{i+i'} z_j'$ for $i+i'+j' \le i+j \le k$.
    Thus, $(i+i', j') \in H_{k-1}$, which means $z_k$ is a linear combination of $\tilde{v}_k$ and $\{A^i \tilde{v}_j\}_{H_{k-1}}$, or equivalently, $\{A^i \tilde{v}_j\}_{i+j \le k}$.
\end{proof}

We now construct the rotation matrices $\Usim_k$. First, we define matrix-valued functions $U_k(\cdot)$, for $k=1, \dots, K$, as follows.
\begin{definition}
    For $1 \le k \le K$, the function $U_k(\cdot)$ takes arguments $\{x_{i,j}\}_{H_{k-1}}$, $y_k$, $z_k$, where the vectors $y_k$ and $z_k$ have unit norm and are both orthogonal to the collection $\{x_{i,j}\}_{H_{k-1}}$. 
    
    To define $U_1(\cdot)$: since $H_{0}$ is empty, the first function $U_1$ only takes arguments $y_1, z_1$, and is such that $U_1(y_1, z_1)$ is a deterministic orthogonal matrix that satisfies $U_1(y_1, z_1)^\top y_1 = z_1$. Note that $U_1(\cdot)$ exists because $y_1$ and $z_1$ both have unit norm; for example, we can complete $y_1$ and $z_1$ to orthonormal bases $(y_1,y_2,\dotsc,y_{d})$, $(z_1,z_2,\dotsc,z_{d})$ and take $U_1(y_1,z_1) = \sum_{i=1}^{d} y_i z_i^\top$.

To define $U_k(\cdot)$: $U_k(\{x_{i,j}\}_{H_{k-1}}, y_k, z_k)$ is a deterministic orthogonal matrix that satisfies
\begin{align}
\begin{aligned}\label{eq:uk_func}
    U_k^\top x_{i,j} &= x_{i,j}\,, \qquad \text{for all}~ (i, j) \in H_{k-1}\,, \\
    U_k^\top y_k &= z_k\,.
\end{aligned}
\end{align}
Such a choice of $U_k$ is always possible, because $k^2 < d$, and because $y_k$ and $z_k$ are orthogonal to $x_{i,j}$; for example, we can start with the identity matrix on the subspace spanned by $\{x_{i,j}\}_{H_{k-1}}$ and add to it a sum of outer products formed by completing $y_k$ and $z_k$ to two orthonormal bases of the orthogonal complement.
\end{definition}

Next, we describe how to construct $\Usim_k$. We will define $\Usim_k$ along with an auxiliary sequence $\{\vssim_k\}_{k=1,2,\dotsc,K-1}$. 

\begin{definition}
    We let $\vssim_1 = v_1$, and $\Usim_1 = U_1(\vsim_1, \vssim_1)$. For $k \ge 2$, $\vssim_k$ and $\Usim_k$ are defined recursively as follows:
\begin{align}
\begin{aligned}\label{eq:vs_us}
    \vssim_k &= v_k\bigl(\{(\Usim_{1:(k-1)})^\top A^i \vsim_j\}_{H_{k-1}}\bigr)\\
    \Usim_k &= U_k\bigl(\{(\Usim_{1:(k-1)})^\top A^i \vsim_j\}_{H_{k-1}},\; (\Usim_{1:(k-1)})^\top \vsim_k,\; \vssim_k\bigr)\, .
\end{aligned}
\end{align}
\end{definition}

Intuitively, one can think of $\vssim_k$ as the $k$th vector the simulator thinks the algorithm is querying, and $\Usim_k$ as a rotation that corresponds $\vssim_k$ to the random unit vector known by Block Krylov.

\begin{proposition}
    Each $\Usim_k$ is well-defined.
\end{proposition}

\begin{proof}
    To show that this choice of $\Usim_k$ is possible, we need to check that $(\Usim_{1:(k-1)})^\top \vsim_k$, $\vssim_k$ both have unit norm and are orthogonal to the subspace $S_k$ spanned by $(\Usim_{1:(k-1)})^\top A^i \vsim_j$ for $(i, j) \in H_{k-1}$. They both have unit norm because $\vsim_k$ and $\vssim_k$ are constructed to have unit norm, and inductively we can assume $\Usim_{1:(k-1)}$ is orthogonal. Note that $\vssim_k$ is orthogonal to $S_k$ by our assumption on the function $v_k(\cdot)$, and $(\Usim_{1:(k-1)})^\top \vsim_k$ is also orthogonal to $S_k$ because
\begin{align*}
    \langle (\Usim_{1:(k-1)})^\top A^i \vsim_j, (\Usim_{1:(k-1)})^\top\vsim_k\rangle = \langle A^i \vsim_j, \vsim_k\rangle = 0\, ,
\end{align*}
    where the second line follows from the definition of $\tilde v_k$.
\end{proof}


 We summarize some important additional properties of $\vssim_k$ and $\Usim_k$ in the following lemma.
\begin{lemma}[Properties of the Simulated Sequences]\label{lem:krylov_aux}
The variables $\Usim_k$ and $\vssim_k$ for $k = 1, \dotsc, K$ defined above satisfy the following properties:
\begin{enumerate}[label=(P\arabic*), ref=(P\arabic*)]
\item \label{prop1} $\vssim_k$ depends only on $\{A^i \vsim_j\}_{H_{k-1}}$, and $\Usim_k$ depends only on $\{A^i \vsim_j\}_{i+j \le k}$.
\item \label{prop2} 
For any $k \ge j$, we have
\begin{align*}
    \vsim_j = \Usim_{1:k} \vssim_j\, .
\end{align*}
\item \label{prop3} For $k \ge 2$, $\vssim_k$ satisfies
\begin{align*}
    \vssim_k
    &= v_k\bigl(\{(\Usim_{1:(k-1)})^\top A^i \Usim_{1:(k-1)} \vssim_j\}_{H_{k-1}}\bigr)\,.
\end{align*}
\item \label{prop4} For $k \ge 2$, $\Usim_k$ satisfies
\begin{align*}
    \Usim_k &= U_k\bigl(\{(\Usim_{1:(k-1)})^\top A^i \Usim_{1:(k-1)}\vssim_j\}_{H_{k-1}},\; (\Usim_{1:(k-1)})^\top \vsim_k,\; \vssim_k\bigr)\, .
\end{align*}
\end{enumerate}
\end{lemma}

Before we present the proof, we highlight the importance of~\ref{prop2} for $k = K$, which roughly states that $(\Usim_{1:K})^\top$ actually sends each Block Krylov-generated vector $\vsim_j$ to the simulated vector $\vssim_j$.

\begin{proof}
    \ref{prop1} is immediate from the definitions, since $\{(i, j): i+j \le k\} = H_{k-1} \cup \{(0, k)\}$. 
    
    To show \ref{prop2}, note that the second property of the function $U_k$ from \eqref{eq:uk_func} implies that
    \begin{align}\label{eq:tildev}
         \vssim_j = (\Usim_j)^\top (\Usim_{1:(j-1)})^\top\vsim_j = (\Usim_{1:j})^\top\vsim_j\, . 
    \end{align}
    This proves~\ref{prop2} for $k = j$. To prove~\ref{prop2} for $k > j$, we use induction on $k$. If~\ref{prop2} holds for $k-1 \ge j$, then
    \begin{align}
        (\Usim_{1:k})^\top \vsim_j = (\Usim_k)^\top (\Usim_{1:(k-1)})^\top \vsim_j = (\Usim_{1:(k-1)})^\top \vsim_j = \vssim_j.
    \end{align}
    Above, the middle equality holds by the first property of \eqref{eq:uk_func}, since $\Usim_k$ fixes $(\Usim_{1:(k-1)})^\top \vsim_j$ because $j \le k-1$. The final equality holds by our inductive hypothesis. So, \ref{prop2} holds for $k$.
    
    Finally, \ref{prop3} and \ref{prop4} then follow from~\ref{prop2}, since $k-1 \ge j$ if $j \in H_{k-1}$. 
\end{proof}

Before proving Lemma~\ref{lem:chen_block_krylov}, we must make one more basic definition.

\begin{definition}
    For $k \ge 2$, given the matrix $A$ and a set $\{v_j\}_{1 \le j \le k-1}$, define $\mathfrak{C}_k$ as the function that satisfies $\mathfrak{C}_k(A, \{v_j\}_{1 \le j \le k-1}) = \{A^i v_j\}_{H_{k-1}}$. In addition, define $\mathfrak{D}_{k} = v_k \circ \mathfrak{C}_k$.
\end{definition}

\medskip

We are now ready to prove Lemma~\ref{lem:chen_block_krylov}.
Although the proof is notationally burdensome, the message is that we can show the equality of distributions inductively by repeatedly invoking the conditioning lemma (Lemma~\ref{lem:conditioning_lemma}), which is designed precisely for the present situation.

\begin{proof}[Proof of Lemma~\ref{lem:chen_block_krylov}]

It is clear that $\vsim_1, \dots, \vsim_K$ and $\Usim_1, \dots, \Usim_K$ satisfy Property 1 in Lemma \ref{lem:chen_block_krylov}. 
We focus on proving the second property.
For $1 \le k \le K$, let $A_{k} \deq (\Usim_{1:k})^\top A \Usim_{1:k}$.
Since we can write $(\Usim_{1:k})^\top A^i \vsim_j = (\Usim_{1:k})^\top A^i (\Usim_{1:k}) \vssim_j = A_k^i \vssim_j$ for any $k \ge j$ by~\ref{prop2} of Lemma \ref{lem:krylov_aux}, it suffices to inductively prove that for all $1 \le k \le K$, 
    \begin{align} \label{eq:goal2}
        (A_k, \{\vssim_j\}_{1 \le j \le k}) \eqdist (A, \{\valg_j\}_{1 \le j \le k}).
    \end{align}

    For the base case of $k = 1$, it suffices to show that $(A_1, \vssim_1) \eqdist (A, \valg_1)$. Note, however, that $\vssim_1 = \valg_1 = v_1$, and $A_1 = (\Usim_1)^\top A (\Usim_1) = U_1(\vsim_1, v_1)^\top A U_1(\vsim_1, v_1)$.
    Since $v_1$ is a deterministic vector, $\vsim_1$ is independent of $A$, and the distribution of $A$ is rotationally invariant, the claim follows.

    For the inductive step, assume we know $(A_k, \{\vssim_j\}_{1 \le j \le k}) \eqdist (A, \{\valg_j\}_{1 \le j \le k})$. 
    Then, note that $\valg_{k+1} = v_{k+1}(\{A^i \valg_j\}_{H_k})$ and $\vssim_{k+1} = v_k(\{A_k^i \vssim_j\}_{H_k})$. Thus, we have $\valg_{k+1} = \mathfrak{D}_{k+1}(A, \{\valg_j\}_{1 \le j \le k})$ and $\vssim_{k+1} = \mathfrak{D}_{k+1}(A_k, \{\vssim_j\}_{1 \le j \le k})$. In addition, because $\Usim_{k+1}$ fixes $A_k^i \vssim_j$ for all $(i, j) \in H_k$ by~\ref{prop4}, we also have that $A_{k+1}^i \vssim_j = A_k^i \vssim_j$ for all $(i, j) \in H_k$, which means $\vssim_{k+1} = \mathfrak{D}_{k+1}(A_{k+1}, \{\vssim_j\}_{1 \le j \le k})$. Therefore, it suffices to show 
    \begin{align} \label{eq:goal3}
        (A_{k+1}, \{\vssim_j\}_{1 \le j \le k})
        \eqdist (A, \{\vssim_j\}_{1 \le j \le k})\,,
    \end{align}
    as this implies $(A_{k+1}, \{\vssim_j\}_{1 \le j \le k+1}) \eqdist (A, \{\valg_j\}_{1 \le j \le k+1})$, which completes the inductive step.

    Next, we show that $\Usim_{k+1}$ sends $\vsim_{k+1}$ to a random unit vector orthogonal to the simulated queries so far.
    Note that $A_{k+1} = (\Usim_{k+1})^\top A_k (\Usim_{k+1})$, where,  by~\ref{prop4},
    \begin{equation} \label{eq:Usim_k+1_redefined}
        \Usim_{k+1} = U_{k+1}(\{A_k^i \vssim_j\}_{H_k}, (\Usim_{1:k})^\top\tilde{v}_{k+1}, \vssim_{k+1}).
    \end{equation}
   Note that $\tilde{v}_{k+1}$ is a random unit vector orthogonal to $\{A^i z_j\}_{H_k}$, or equivalently, it is a random unit vector orthogonal to $\{A^i \vsim_j\}_{H_k}$. However, since $(\Usim_{1:k})^\top A^i \vsim_j = (\Usim_{1:k})^\top A^i(\Usim_{1:k}) \vssim_j = A_k^i \vssim_j$ for all $(i, j) \in H_k$ (by~\ref{prop2}), this means that $(\Usim_{1:k})^\top \vsim_{k+1}$ is orthogonal to $\{A_k^i \vssim_j\}_{H_k}$. In addition, by~\ref{prop1} and the definition of $\vssim_k, \Usim_k$ (Equation \eqref{eq:vs_us}), we have that the first and third arguments of $\Usim_{k+1}$ only depend on $\{A^i \tilde{v}_j\}_{H_k}.$ Thus, the random direction of $\tilde{v}_{k+1}$ has no dependence on $\{A^i \tilde{v}_j\}_{H_k}$ apart from being orthogonal to them, which means by~\ref{prop1}, $(\Usim_{1:k})^\top \vsim_{k+1}$ is a \emph{uniformly random} unit vector orthogonal to $\{A_k^i \vssim_j\}_{H_k}$.
    
    Recalling that $\vssim_{k+1} = \mathfrak{D}_{k+1}(A_{k}, \{\vssim_j\}_{1 \le j \le k})$, this means that we can rewrite \eqref{eq:Usim_k+1_redefined} as
    \begin{align}
        \Usim_{k+1} =& \hspace{0.1cm} U_{k+1}\left(\{A_k^i \vssim_j\}_{H_k}, \hat{v}^{\msf{sim}}, \mathfrak{D}_{k+1}(A_{k}, \{\vssim_j\}_{1 \le j \le k})\right), \label{eq:uk+1-sim-new}\\
    \intertext{where $\hat{v}^{\msf{sim}}$ is a random unit vector orthogonal to $\{A_k^i \vssim_j\}_{H_k}$.
    As a result, if we define}
        U_{k+1}^{\msf{alg}} \deq& \hspace{0.1cm} U_{k+1}\left(\{A^i \valg_j\}_{H_k}, \hat{v}^{\msf{alg}}, \mathfrak{D}_{k+1}(A, \{\valg_j\}_{1 \le j \le k})\right), \label{eq:uk+1-alg}
    \end{align}
    where $\hat{v}^{\msf{alg}}$ is a random unit vector orthogonal to $\{A^i \valg_j\}_{H_k}$, then 
    \begin{align*}
        (A_{k+1}, \{\vssim_j\}_{1 \le j \le k}) &= ((\Usim_{k+1})^\top A_k (\Usim_{k+1}), \{\vssim_j\}_{1 \le j \le k}) \\
        &\eqdist \left((U_{k+1}^{\msf{alg}})^\top A (U_{k+1}^{\msf{alg}}), \{\valg_j\}_{1 \le j \le k}\right).
    \end{align*}
    Above, the first equality follows by definition, and the third follows from our inductive hypothesis that $(A_k, \{\vssim_j\}_{1 \le j \le k}) \eqdist (A, \{\valg_j\}_{1 \le j \le k})$, along with \eqref{eq:uk+1-sim-new} and \eqref{eq:uk+1-alg}.

    We are now in a position to apply the conditioning lemma (Lemma~\ref{lem:conditioning_lemma}).
    Note that $U_{k+1}^{\msf{alg}}$ only depends on $\{A^i \valg_j\}_{H_k}$ (as well as some randomness in $\hat{v}^{\msf{alg}}$, but the randomness is independent of everything else given $\{A^i \valg_j\}_{H_k}$, so we can safely condition on it). Hence, we can apply the conditioning lemma with $U_{k+1}^{\msf{alg}}$, to obtain that 
    \begin{align*}
        (A_{k+1}, \{\vssim_j\}_{1 \le j \le k}) \eqdist \left((U_{k+1}^{\msf{alg}})^\top A (U_{k+1}^{\msf{alg}}), \{\valg_j\}_{1 \le j \le k}\right) \eqdist \left(A, \{\valg_j\}_{1 \le j \le k}\right) \,,
    \end{align*}
    which establishes \eqref{eq:goal3} and thereby concludes the proof.
\end{proof}